\documentclass[sigconf]{acmart}

\usepackage{adjustbox}
\usepackage{amssymb}
\usepackage{balance}
\usepackage{bm}
\usepackage{caption}
\usepackage{color}
\usepackage{makecell}
\usepackage{mathtools}
\usepackage{microtype}
\usepackage{multirow}
\usepackage{paralist}
\usepackage{subcaption}
\usepackage{tabularx}
\usepackage{url}
\usepackage{xcolor}

\usepackage{booktabs}

\PassOptionsToPackage{hyphens}{url}\usepackage{hyperref} 
\definecolor{mylinkcolor}{RGB}{0,0,140}
\hypersetup{colorlinks,allcolors=mylinkcolor,citecolor=mylinkcolor}
\usepackage[capitalize]{cleveref}
\crefname{figure}{Figure}{Figures}

\newcommand{\xhdr}[1]{\vspace{0.5mm}\noindent{{\bf #1.}}\hspace{0.5mm}}
\newcommand{\xhdrq}[1]{\vspace{0.5mm}\noindent{{\bf #1?}}\hspace{0.5mm}}

\newcommand{\phantomsubfigure}[1]{\begin{subfigure}[b]{0.1\textwidth}\phantomcaption\label{#1}\end{subfigure}}

\newcommand{\proxyX}{\hat{X}}
\newcommand{\proxyXparam}[1]{\hat{X}_{#1}}
\newcommand{\proxyY}{\hat{Y}}
\newcommand{\proxyYparam}[1]{\hat{Y}_{#1}}

\newcommand{\proxyZparam}[1]{Z_{#1}}

\newcommand{\prob}[1]{\textnormal{Pr}[#1]}
\newcommand{\expec}[1]{\mathbb{E}[#1]}
\newcommand{\var}[1]{\mathbb{V}[#1]}
\newcommand{\snr}[1]{\textnormal{SNR}[#1]}
\newcommand{\given}{\;\vert\;}
\newcommand{\edgetrain}{E^{\textnormal{train}}}
\newcommand{\edgetest}{E^{\textnormal{test}}}
\newcommand{\graphtrain}{G^{\textnormal{train}}}
\newcommand{\edgetrainparam}[1]{E_{#1}^{\textnormal{train}}}

\newcommand{\nodeparam}[1]{V_{#1}}

\begin{document}

\title{Link Prediction in Networks with Core-Fringe Data}

\author{Austin R.~Benson}
\affiliation{%
  \institution{Cornell University}
}
\email{arb@cs.cornell.edu}

\author{Jon Kleinberg}
\affiliation{%
  \institution{Cornell University}
}
\email{kleinber@cs.cornell.edu}


\begin{abstract}
Data collection often involves the 
partial measurement of a larger system. 
A common example arises in collecting network data: 
we often obtain network datasets by
recording all of the interactions among a small set of core nodes, so that we
end up with a measurement of the network consisting of these core nodes 
along with a potentially much larger set of fringe nodes that have
links to the core. Given the ubiquity of this process for assembling
network data, it is crucial to understand the role of such a
``core-fringe'' structure.

Here we study how the inclusion of fringe nodes affects the standard
task of network link prediction. 
One might initially think the inclusion of any additional data is useful, and hence that
it should be beneficial to include all fringe nodes that are available. However,
we find that this is not true; in fact, there is substantial variability in
the value of the fringe nodes for prediction. Once an algorithm is selected,
in some datasets, including any additional data from the fringe can
actually hurt prediction performance; in other datasets, including some amount
of fringe information is useful before prediction performance saturates or even
declines; and in further cases, including the entire fringe leads to the best performance.
While such variety might seem surprising, we show
that these behaviors are exhibited by simple random graph models.
\end{abstract}

\maketitle

\section{Introduction}

In a wide range of data analysis problems, the underlying data 
typically comes from partial measurement of a larger system.
This is a ubiquitous issue in the study of networks, where 
the network we are analyzing is almost always embedded in some larger
surrounding network 
\cite{Kim-2011-completion,Kossinets-2006-missing,Laumann-1989-boundary}.
Such considerations apply to systems at all scales.
For example, when studying the communication network of an organization,
we can potentially gain additional information
if we know the structure of employee interactions
with people outside the organization as well~\cite{Romero-2016-stress}.
A similar issue applies to large-scale systems.
If we are analyzing the links within a large online social network,
or the call traffic data from a large telecommunications provider,
we could benefit from knowing the interactions that members
of these systems have with individuals who are not part of the platform,
or who do not receive service from the provider.

\begin{figure}
\includegraphics[width=0.35\columnwidth]{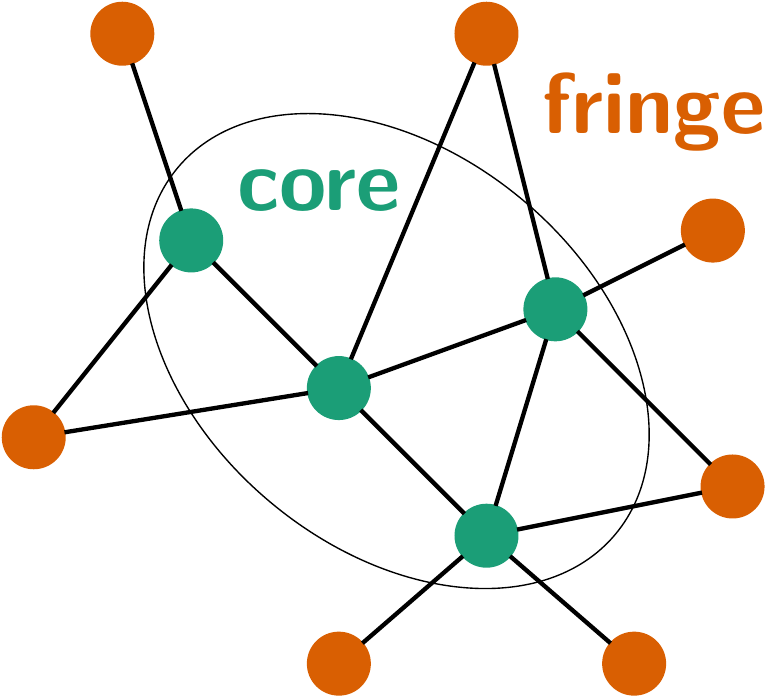}
\hskip30pt
\includegraphics[width=0.35\columnwidth]{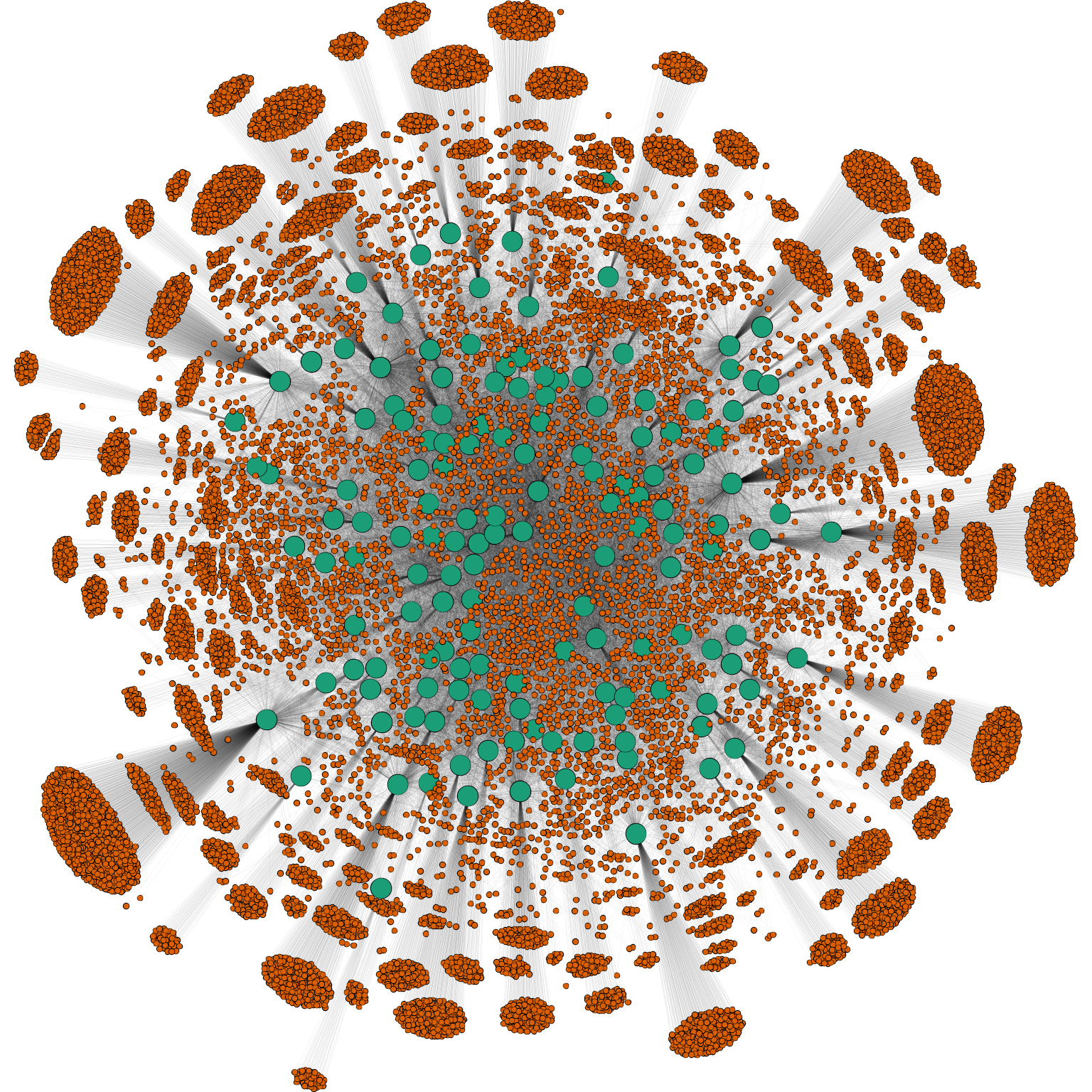}
\caption{Core-fringe structure.
{\bf (Left)} 
Illustrative network with labeled core-fringe structure (core nodes in green and fringe nodes in orange). 
We observe all of the links involving core nodes (in green). 
Each edge is between two core nodes or between one core and one fringe node.
{\bf (Right)}
Core-fringe structure in the Enron email network, which results from data 
collection---the core nodes in green correspond to the roughly 150 accounts whose
emails were released as part of a federal investigation \cite{Klimt-2004-Enron}.
Here, the number of fringe nodes is orders of magnitude larger than the number of core nodes.
}
\label{fig:cf}
\end{figure}

Network data can therefore be viewed as having a {\em core-fringe}
structure (following the terminology of our previous work~\cite{Benson-2018-found-graph-data}):
we collect data by measuring all of the interactions involving a
{\em core} set of nodes, and in the process we also learn about the
core's interaction with a typically larger set of additional nodes---the
{\em fringe}---that does not directly belong to the measured system.
\Cref{fig:cf} illustrates the basic structure schematically
and also in a typical real-life scenario: If we collect a
dataset by measuring all email communication to and from
the executives of a company (the core), then the data will also
include links to people outside this set with whom
members of the core exchanged email (the fringe).%
\footnote{This distinction between core and fringe is
fundamentally driven by measurement of the available data;
we have measured all interactions involving members of
the core, and this brings the fringe indirectly into the data.
As such, it is distinct from work on the {\em core-periphery structure}
of networks,
which typically refers to settings in which the core and periphery
both fully belong to the measured network, and the distinction is
in the level of centrality or status that the core has relative to
the periphery 
\cite{Borgatti-2000-CP,Holme-2005-CP,Rombach-2017-CP,Zhang-2015-SBM-CP}.}
We thus have two kinds of links:
links between two members of the core, and links between a member of the
core and a member of the fringe. Links between fringe members are 
not visible, even though we are aware of both fringe nodes through their interactions with the core.

Despite the fundamental role of core-fringe structure in network data
and a long history of awareness of this issue in the
social sciences \cite{Laumann-1989-boundary},
there has been little systematic attention paid to its 
implications in basic network inference tasks.
If we are trying to predict network structure on 
a measured set of core nodes,
what is the best way to make use of the fringe?
Is it even clear that incorporating the fringe nodes will help?
To study these questions, it is important to have
a concrete task on the network where notions of performance
as a function of available data are precise.

\xhdr{The present work: Core-fringe link prediction}
In this paper, we study the role of core-fringe structure 
through one of the standard network inference
problems: {\em link prediction}~\cite{LibenNowell-2007-link-pred,Lu-2011-link-pred}.
Link prediction is a problem in which the goal is
to predict the presence of {\em unseen} links in a network.
Links may be unseen for a variety of reasons, depending on the
application---we may have observed the network up to a certain
point in time and want to forecast new links, or we may have 
collected a subset of the links and want to know which additional ones are present.

Abstractly, we will think of the link prediction problem as 
operating on a graph $G = (V,E)$ whose edges are divided
into a set of observed edges and a set of unseen edges.
From the network structure on the observed edges, we would like
to predict the presence of the unseen edges as accurately as possible.
A large range of heuristics have been proposed for this problem,
many of them based on the empirical principle that nodes with
neighbors in common are generally more likely to be connected by a link
\cite{LibenNowell-2007-link-pred,Lu-2011-link-pred,Rapoport-1953-triadic}.

The issue of core-fringe structure shows up starkly in the
link prediction problem. Suppose the graph $G$ has nodes that are divided into
a core set $C$ and a fringe set $F$, and our goal is to predict unseen
links between pairs of nodes in the core.
One option would be to perform this task using only the portion of
$G$ induced on the core nodes.
But we could also perform the task using larger amounts of $G$ by
taking the union of the core nodes with any subset of the fringe,
or with all of the fringe.
The key question is how much of the fringe we should include 
if our goal is to maximize performance on the core;
existing work provides little guidance about this question.

\xhdrq{How much do fringe nodes help in link prediction}
We explore this question in a broad collection of network datasets derived
from email, telecommunication, and online social networks.
For concreteness, our most basic formulation 
draws on common-neighbor heuristics to answer the following version of the
link prediction question: given two pairs of nodes drawn from the core,
$\{u,v\}$ and $\{w,z\}$, which pair is more likely to be connected by a link?
(In our evaluation framework, we will focus on cases in which 
exactly one of these pairs is truly connected by a link,
thus yielding a setting with a clear correct answer.)
To answer this question, we could use information about the common
neighbors that $\{u,v\}$ and $\{w,z\}$ have only in the core, or 
also in any subset of the fringe.
How much fringe information should we use, if we want to maximize
our probability of getting the correct answer?

It would be natural to suspect
that using all available data, i.e., including all of the fringe nodes,
would maximize performance.
What we find, however, is a wide range of behaviors.
In some of our domains---particularly the social-networking data---link
prediction performance increases monotonically in the amount
of fringe data used, though with diminishing returns as 
we incorporate the entire fringe.
In the other domains, however, we find a number of instances where
using an intermediate level of fringe, i.e.,  a proper subset of
the fringe nodes, yields a performance that dominates
the option of including all of the fringe or none of it.
And there are also cases where prediction is best when 
we ignore the fringe entirely.
Given that proper subsets of the fringe can yield better performance
than either extreme, we also consider the process of selecting a subset
of the fringe; in particular, we study different natural {\em orderings}
of the fringe nodes and then select a subset by
searching over prefixes of these orderings.

To try understanding this diversity of results, we turn to
basic random graph models, adapting them to capture the problem of
link prediction in the presence of core-fringe structure.
We find that simple models are rich enough to display the same diversity of 
behaviors in performance, where the optimal amount of fringe might be 
all, some, enough, or none.
More specifically, we analyze the signal-to-noise ratio for our basic link prediction primitive
in two heavily-studied network models: stochastic block models (SBMs), in which
random edges are added with different probabilities between a set of
planted clusters~\cite{Abbe-2016-exact,Abbe-2018-community,Decelle-2011-SBM,Mossel-2014-belief};
and small-world lattice models, in which 
nodes are embedded in a lattice, and links between nodes
are added with probability decaying as a power of the distance~\cite{Watts-1998-small-world,Kleinberg-2006-ICM}.
We prove that there are instances of the SBM with certain linking probabilities
in which the signal-to-noise ratio is optimized by including all
the fringe, enough of the fringe, or none of the fringe.
For small-world lattice models, we find in the most basic formulation
that the signal-to-noise ratio is optimized by including
an intermediate amount of fringe: essentially, if the core is a bounded
geometric region in the lattice, then the optimal strategy for link 
prediction is to include the fringe in a larger region that extends
out from the core; but if we grow this region too far then
performance will decline.

The analysis of these models provides us with some qualitative
higher-level insight into the role of fringe nodes in link prediction.
In particular, the analysis can be roughly summarized as follows:
the fringe nodes that are most well-connected to the core
are providing valuable predictive signal without significantly 
increasing the noise; but as we continue including fringe nodes that
are less and less well-connected to the core, the signal decreases much faster
than the noise, and eventually the further fringe nodes are 
primarily adding noise in a way that hurts prediction performance.

More broadly, the results here indicate that the question of how
to handle core-fringe structure in network prediction problems
is a rich subject for investigation, and an important one given
how common this structure is in network data collection.
An implication of both our empirical and theoretical results is that
it can be important for problems such as link prediction to measure
performance with varying amounts of additional data, and to accurately
evaluate the extent to which this additional data is primarily
adding signal or noise to the underlying decision problem.

\section{Empirical network analysis}\label{sec:empirical}

We first empirically study how including fringe nodes can affect
link prediction on a number of datasets.
While we might guess that any additional data
we can gather would be useful for prediction, we see that
this is not the case.
In different datasets, incorporating all, none, some, or enough fringe
data leads to the best performance.
We then show in \cref{sec:theory} that this
variability is also exhibited in the behavior of simple random graph models.

\xhdr{Evaluating link prediction with core-fringe structure}
There are several ways to evaluate link prediction performance~\cite{LibenNowell-2007-link-pred,Lu-2011-link-pred}.
We set up the prediction task in a natural way that
is also amenable to theoretical analysis in \cref{sec:theory}.
We assume that we have a graph $G = (V, E)$ with a known set of core nodes $C \subseteq V$
and fringe nodes $F = V - C$.
The edge set is partitioned into $\edgetrain$ and $\edgetest$, where 
$\edgetest$ is a subset of the edges that connect two nodes in the core $C$.
The form of $\edgetest$ depends on the dataset,
which we describe in the following sections.
In general, our core-fringe link prediction evaluation is based on how well
we can predict elements of $\edgetest$ given
the graph $\graphtrain = (V, \edgetrain)$.

Our atomic prediction task 
considers two pairs of nodes $\{u, v\}$ and $\{w, z\}$ such that
(i) all four nodes are in the core (i.e., $u, v, w, z \in C$);
(ii) neither pair is an edge in $\edgetrain$;
(iii) the edge $(u, v)$ is a positive sample, meaning that $(u, v) \in \edgetest$; and
(iv) the edge $(w, z)$ is a negative sample, meaning that $(w, z) \notin \edgetest$.
We use an algorithm that takes as input $\graphtrain$ and outputs a score
function $s(x, y)$ for any pair of nodes $x, y \in C$; 
the algorithm then predicts that the
pair of nodes with the higher score is more likely to be in the test set.
Thus, the algorithm makes a correct prediction if $s(u, v) > s(w, z)$.
We sample many such 4-tuples of nodes uniformly at random and measure
the fraction of correct predictions.

We evaluate two score functions that are common heuristics for link
prediction~\cite{LibenNowell-2007-link-pred}. The first is the \emph{number of common neighbors}:
\begin{equation}
s(x, y) = \lvert N(x) \cap N(y) \rvert,
\end{equation}
where $N(z)$ is the set of neighbors of node $z$ in the graph.
The second is the \emph{Jaccard similarity} of the neighbor sets:
\begin{equation}
s(x, y) = \frac{\lvert N(x) \cap N(y) \rvert}{\lvert N(x) \cup N(y) \rvert}.
\end{equation}
We choose these score functions for a few reasons.
First, they are flexible enough to be feasibly deployed even if only
minimal information about the fringe is available;
more generally, their robustness 
has motivated their use as heuristics in practice~\cite{Goel-2013-Discovering,Gupta-2013-WTF} and throughout
the line of research on link prediction~\cite{LibenNowell-2007-link-pred}.
Second, 
they are amenable to analysis: we can explain some of our results by analyzing
the common neighbors heuristic on random graph models, and
they are rich enough to expose a complex landscape of
behavior. This is sufficient for the present study, but it would be
interesting to examine more sophisticated link prediction algorithms
in our core-fringe framework as well.

We parameterize the training data 
by how much fringe information is included. To do this, we construct
a nested sequence of sets of vertices, each of which induces a set of
training edges. Specifically, the initial set of vertices is the core, and
we continue to add fringe nodes to construct a nested sequence of vertices:
\begin{equation}\label{eq:nodeseq}
C = \nodeparam{0} \subseteq \nodeparam{1} \subseteq \cdots \subseteq \nodeparam{D} = V.
\end{equation}
The nested sequence of vertex sets then induces a nested sequence
of edges that are the training data for the link prediction algorithm;
for a value of $d$ between $0$ and $D$, we write
\begin{equation}\label{eq:edgetrainseq}
\edgetrainparam{d} = \{ (u, v) \in E \;\vert\; u, v \in \nodeparam{d} \} \cap \edgetrain.
\end{equation}
From \cref{eq:nodeseq,eq:edgetrainseq},
$\edgetrainparam{d} \subseteq \edgetrainparam{d+1}$, and
$\edgetrainparam{D} = \edgetrain$.
The parameterization of the vertex sets will depend on the dataset, and
we examine multiple sequences $\{\nodeparam{d}\}$ to study how
different interpretations of the fringe give varying outcomes.
Our main point of study is link prediction performance \emph{as a function of $d$}.

\subsection{Email networks}

\begin{table}[tb]
\setlength{\tabcolsep}{5pt}
\centering
\caption{Summary statistics of email datasets.}
\begin{tabular}{r c c c c}
\toprule
        & \# core & \# fringe & \# core-core & \# core-fringe \\
Dataset & nodes   & nodes     & edges & edges \\
\midrule
email-Enron      & 148 & 18,444 & 1,344 & 41,883 \\
email-Avocado & 256 & 27,988 & 7,416 & 50,048 \\
email-Eu            & 1,218 & 200,372 & 16,064 & 303,793 \\
email-W3C        & 1,995 & 18,086 & 1,777 & 30,097 \\
\midrule
email-Enron-1 & 37 & 7,511 & 86 & 11,862 \\
email-Enron-2 & 37 & 6,440 & 81 & 10,648 \\
email-Enron-3 & 37 & 6,379 & 80 & 10,390 \\
email-Enron-4 & 37 & 6,587 & 95 & 10,987 \\
\bottomrule
\end{tabular}
\label{tab:email_datasets}
\end{table}

\begin{figure*}
\phantomsubfigure{fig:perf_email_A}
\phantomsubfigure{fig:perf_email_B}
\phantomsubfigure{fig:perf_email_C}
\phantomsubfigure{fig:perf_email_D}
\phantomsubfigure{fig:perf_email_E}
\phantomsubfigure{fig:perf_email_F}
\phantomsubfigure{fig:perf_email_G}
\phantomsubfigure{fig:perf_email_H}
\includegraphics[width=2\columnwidth]{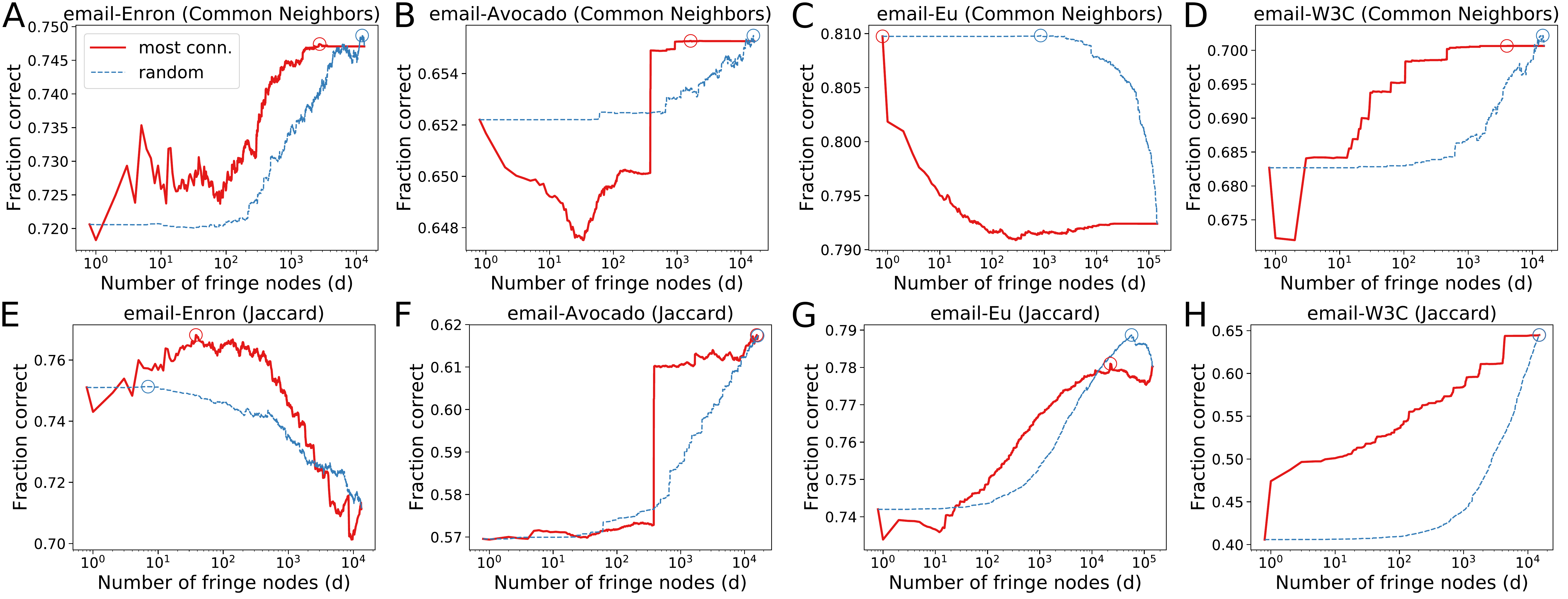}
\caption{Link prediction performance of the Common Neighbors (top) and Jaccard similarity (bottom) score functions
on four email networks as a function of $d$, the number of fringe nodes included. 
Two orderings of fringe nodes are considered: one by the most connections
to the core (red) and one random (blue). A circle marks the best performance.
There is a striking variety in how the fringe affects performance.
In some cases, we should ignore the fringe entirely (C);
in others, performance increases with the size of the fringe (F, H);
and in still others, some intermediate amount of fringe is optimal (E, G).}
\label{fig:perf_email}
\end{figure*}

\begin{figure*}
\phantomsubfigure{fig:perf_enron_A}
\phantomsubfigure{fig:perf_enron_B}
\phantomsubfigure{fig:perf_enron_C}
\phantomsubfigure{fig:perf_enron_D}
\phantomsubfigure{fig:perf_enron_E}
\phantomsubfigure{fig:perf_enron_F}
\phantomsubfigure{fig:perf_enron_G}
\phantomsubfigure{fig:perf_enron_H}
\includegraphics[width=2\columnwidth]{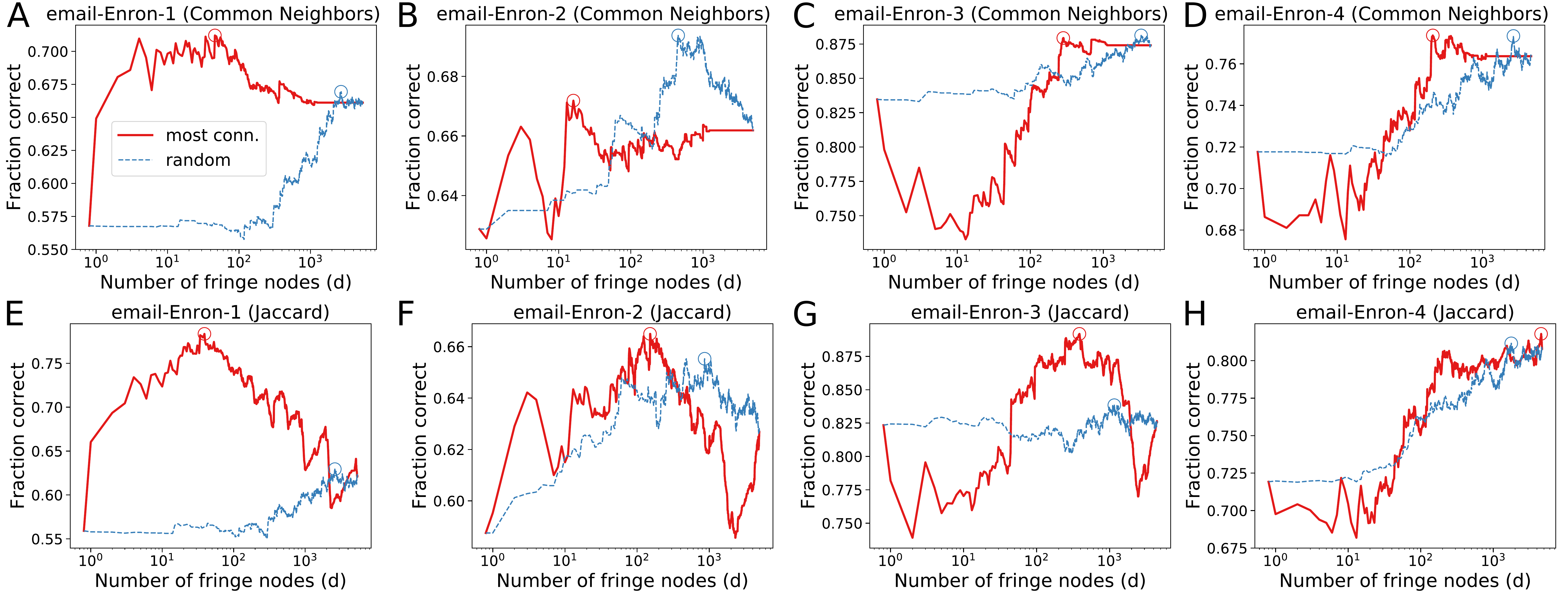}
\caption{Link prediction performance experiments analogous to those in \cref{fig:perf_email} but on four
subsets of email-Enron. In most cases, including some interior amount of fringe nodes---between 10 and a few hundred---yields the optimal performance.}
\label{fig:perf_enron}
\end{figure*}

Our first set of experiments analyzes email networks.
The core nodes in these datasets are members of some organization,
and the fringe nodes are those outside of the organization that communicate
with those in the core. We use four email networks; in each, the
nodes are email addresses, and the time that an edge formed
is given by the timestamp of the first email between two nodes.
For simplicity, we consider all graphs to be undirected, even though there
is natural directionality in the links.
Thus, each dataset is a simple, undirected graph, where each edge has a timestamp
and each node is labeled as core or fringe.
The four datasets are
(i) \emph{email-Enron}: the network in \cref{fig:cf},
  where the core nodes correspond to accounts whose emails
  were released as part of a federal investigation~\cite{Klimt-2004-Enron,Benson-2018-found-graph-data};
(ii) \emph{email-Avocado}: the email network of a now-defunct company, where the core
  nodes are employees (we removed accounts associated with non-people, such as conference rooms).\footnote{\url{https://catalog.ldc.upenn.edu/LDC2015T03}}  
(iii) \emph{email-Eu}: a network that consists of emails involving members of a European
  research institution, where the core nodes are the institution's members~\cite{Leskovec-2007-densification,Yin-2017-local}; and
(iv) \emph{email-W3C}: a network from W3C email threads, where
  core nodes are those addresses with a \texttt{w3.org} domain~\cite{Craswell-2005-TREC,Benson-2018-found-graph-data}.
\cref{tab:email_datasets} provides basic summary statistics.

An entire email network dataset is a graph $G = (V, E)$, where $C \subseteq V$
is a set of core nodes, and each edge $e \in E$
is associated with a timestamp $t_e$.
Here, our test set  is derived from the temporal information.
Let $t^*$ be the 80th percentile of timestamps on edges between core nodes.
Our test set of edges is the final 20\% of edges between core nodes that
appear in the dataset:
\begin{equation}
\edgetest = \{(u, v) \in E \;\vert\; u, v \in C \text{ and } t_{(u, v)} \ge t^*\}.
\end{equation}
The training set is then given by edges appearing before $t^*$,
i.e., the edges appearing in the first 80\% of time spanned by the data:
\begin{equation}\label{eq:email_train}
\edgetrain = \{(u, v) \in E \;\vert\; t_{(u, v)} < t^*\}.
\end{equation}

\xhdr{Fringe ordering}
Next, we form the nested sequence of training set edges (\cref{eq:edgetrainseq})
by sequentially increasing the amount of fringe information (\cref{eq:nodeseq}). 
Recall that $\edgetrainparam{d}$ is simply the set of edges
in $\edgetrain$ in which both end points are in the vertex set $\nodeparam{d}$.
To build $\{ \nodeparam{d} \}$, we start with $\nodeparam{0} = C$, the core set,
and then add the fringe nodes one by one in order of decreasing
degree in the graph $\graphtrain = (V, \edgetrain)$.
By definition, fringe nodes cannot link between themselves, so this ordering is equivalent
to adding fringe nodes in order of the number of core nodes to which they connect.
For purposes of comparison, we also evaluate
this ordering relative to a uniformly random ordering of the fringe nodes. 
To summarize, given an ordering $\sigma$ of the fringe nodes $F$,
we create a nested sequence of node sets $\nodeparam{d} = C \cup \{\sigma_1, \ldots, \sigma_d\}$
in two ways:
(i) \emph{Most connected}: $\sigma$ is the ordering of nodes in the fringe $F$ by decreasing degree in the graph induced by $\edgetrain$ (\cref{eq:email_train}); and
(ii) \emph{Random}: $\sigma$ is a random ordering of the nodes in $F$. 

\xhdr{Link prediction}
We use the \emph{most connected} and \emph{random} ordering to predict links in the test set of edges, 
as described at the beginning of \cref{sec:empirical}.
Recall that we needed a set of candidate comparisons between two potential edges (one of which
does appear in the test). We guess that the pair of nodes with the larger number of common neighbors
or larger Jaccard similarity score will be the set that appears in the test set.
We sample $10 \cdot \lvert \edgetest \rvert$ pairs from $\edgetest$ (allowing
for repeats) and combine each of them with two nodes selected uniformly at random that never form an edge.
Prediction performance is measured in terms of the fraction
of correct guesses as a function of $d$, the number of fringe nodes included. This
entire procedure is repeated 10 times (with 10 different sets of random samples) and
the mean accuracy is reported in \cref{fig:perf_email}.

The results exhibit a wide variety of behavior. 
In some cases, performance tends to increase monotonically with the
number of fringe nodes (\cref{fig:perf_email_F,fig:perf_email_H}).
In one case, we achieve optimal performance by ignoring the fringe entirely
(\cref{fig:perf_email_C}).
In yet another case, some interior amount of fringe is optimal before prediction
performance degrades (\cref{fig:perf_email_E}).
In several cases, we see a saturation effect, where performance flattens as we increase
more fringe nodes (e.g., \cref{fig:perf_email_A,fig:perf_email_D}). This is partly a consequence of how we ordered the fringe---nodes
included towards the end of the sequence are less connected and thus have relatively less
impact on the score functions. In these cases, one practical consequence is that we could ignore large
amounts of the data and get roughly the same performance, which would save computation time.
In another case, the first few hundred most connected fringe nodes leads to worse performance,
but eventually having enough fringe improves performance (\cref{fig:perf_email_B}).
Finally, there are also cases where the optimal performance over $d$ is better for a random
ordering of the fringe than for the most connected ordering 
(\cref{fig:perf_email_D,fig:perf_email_G}).

We repeated the same set of experiments on subgraphs of email-Enron
by partitioning the core set of nodes $C$ into four groups to induce four
different datasets. In each dataset, the other members of the core are removed
from the graph entirely (the bottom part of \cref{tab:email_datasets} lists basic summary statistics).
The results in \cref{fig:perf_enron} provide further evidence that it is a priori unclear how much fringe
information one should include to achieve the best performance.
In nearly all cases, the optimal performance when including fringe nodes in order
of connectedness is somewhere between 10 and a few hundred nodes, out of a total of 
several thousand. And we again see that including initial fringe information by 
connectedness has worse performance than ignoring the fringe entirely, until eventually incorporating
enough fringe information provides enough signal to improve prediction performance
(\cref{fig:perf_enron_C,fig:perf_enron_D}).

\subsection{Telecommunications networks}

\begin{table}[tb]
\setlength{\tabcolsep}{5pt}
\centering
\caption{Summary statistics of telecommunications datasets. 
Core nodes are participants in the Reality Mining study.}
\begin{tabular}{r c c c c}
\toprule
        & \# core & \# fringe & \# core-core & \# core-fringe \\
Dataset & nodes   & nodes     & edges & edges \\
\midrule
call-Reality      & 91 & 8,927 & 127 & 10,512 \\
text-Reality        & 91 & 1,087 & 32 & 1,920 \\
\bottomrule
\end{tabular}
\label{tab:telecomms_datasets}
\end{table}

Next, we study telecommunications datasets from cell phone usage
amongst individuals participating in the Reality Mining project~\cite{Eagle-2005-Reality}.
This project recorded cell phone activity of students and faculty in the MIT Media Laboratory,
including calls and SMS texts between phone numbers.
We consider the participants (more, specifically, their phone numbers) 
as the core nodes
in our network. Edges are phone calls or SMS texts between two people,
some of which are fringe nodes corresponding to people who were not recruited for the experiment.
We process the data in the same way as for email networks---directionality was removed
and the edges are accompanied by the earliest timestamp of communication between the two nodes.
We study two datasets:
(i) \emph{call-Reality}: the network of phone calls~\cite{Eagle-2005-Reality,Benson-2018-found-graph-data}; and
(ii) \emph{text-Reality}: the network of SMS texts~\cite{Eagle-2005-Reality,Benson-2018-found-graph-data}.
\Cref{tab:telecomms_datasets} provides some basic summary statistics.

\begin{figure}[tb]
\phantomsubfigure{fig:perf_telecomms_A}
\phantomsubfigure{fig:perf_telecomms_B}
\phantomsubfigure{fig:perf_telecomms_C}
\phantomsubfigure{fig:perf_telecomms_D}
\includegraphics[width=\columnwidth]{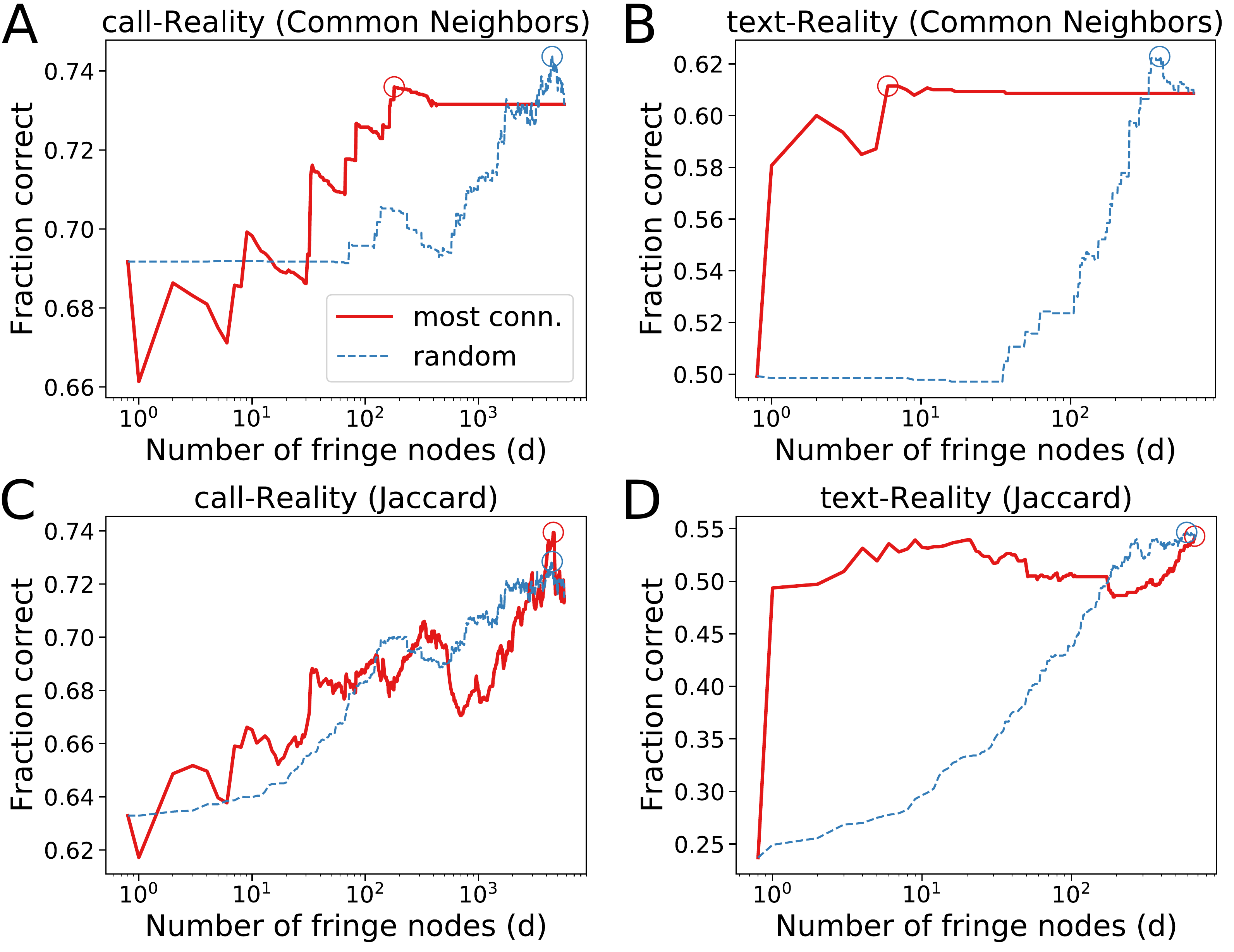}
\caption{Link prediction performance in the telecommunications datasets as a function of $d$,
 the number of fringe nodes used for prediction. Circles mark the largest value.
With the Common Neighbors score function, a small number of fringe nodes
is optimal for these datasets.}
\label{fig:perf_telecomms}
\end{figure}

\xhdr{Fringe ordering}
The structure of these networks is the same
as the email networks---the dataset is a recording of the interactions
of a small set of core nodes with a larger set of fringe nodes.
We use the same two orderings---most connected to core and random---as
we did for the email datasets.
Thus, the nested sequence of node sets $\{\nodeparam{d}\}$ is again constructed by 
adding one node at a time.

\xhdr{Link prediction}
\Cref{fig:perf_telecomms} shows the link prediction performance
on the telecommunications datasets.
With the Common Neighbors score function, we again find that the optimal
amount of fringe is a small fraction of the entire dataset---around 100 of nearly
9,000 nodes  in call-Reality (\cref{fig:perf_telecomms_A})
and around 10 of over 1,000 nodes in text-Reality (\cref{fig:perf_telecomms_B}).
The performance of the Jaccard similarity also has an interior optimum for the
call-Reality dataset, although the optimum size here is larger---around half of the nodes.

Prediction performance with the fringe nodes ordered by connectedness to
the core is in general quite erratic for the call-Reality dataset.
This is additional evidence that
the fringe nodes can be a noisy source of information.
For instance, just including the first fringe node results in a noticeable drop
in prediction performance for both the Common Neighbors and Jaccard score functions.

\subsection{Online social networks}

\begin{table}[b]
\setlength{\tabcolsep}{5pt}
\centering
\caption{Summary statistics of LiveJournal networks. The sets of core nodes
are users in particular states or counties.}
\begin{tabular}{r c c c c}
\toprule
        & \# core & \# fringe & \# core-core & \# core-fringe \\
Dataset & nodes   & nodes     & edges & edges \\
\midrule
Wisconsin      & 16,842 & 58,965 & 48,078 & 87,723 \\
Texas      & 65,617 & 155,357 & 256,174 & 312,746 \\
New York      & 82,275 & 208,516 & 281,981 & 477,845 \\
California      & 152,171 & 244,605 & 712,803 & 722,835 \\
\midrule
Marathon      & 223 & 1,032 & 390 & 1,363 \\
Eau Claire    & 295 & 1,392 & 252 & 1,635 \\
Dane      & 2,281 & 15,580 & 5,192 & 21,156 \\
Milwaukee    & 3,743 & 19,934 & 10,750 & 30,955 \\
\bottomrule
\end{tabular}
\label{tab:lj_datasets}
\end{table}

\begin{figure*}
\phantomsubfigure{fig:perf_lj_states_A}
\phantomsubfigure{fig:perf_lj_states_B}
\phantomsubfigure{fig:perf_lj_states_C}
\phantomsubfigure{fig:perf_lj_states_D}
\phantomsubfigure{fig:perf_lj_states_E}
\phantomsubfigure{fig:perf_lj_states_F}
\phantomsubfigure{fig:perf_lj_states_G}
\phantomsubfigure{fig:perf_lj_states_H}
\includegraphics[width=2\columnwidth]{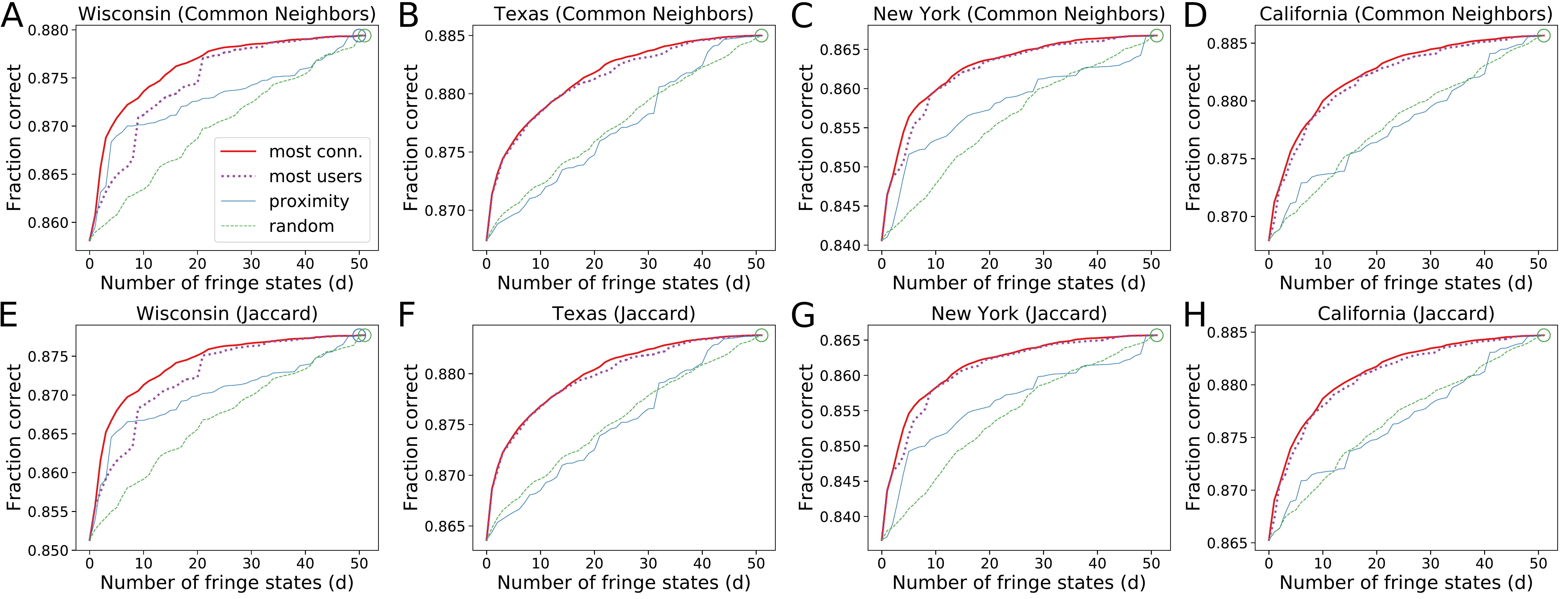}
\caption{Link prediction performance of the Common Neighbors (top) and Jaccard similarity (bottom) score functions
on four LiveJournal datasets where the core consists of users in a state.
Performance is measured as a function of $d$, the number of fringe states included for prediction (for four different orderings of the fringe).
The most connected ordering performs the best, monotonically increases with $d$,
and saturates when $d \ge 20$.}
\label{fig:perf_lj_states}
\end{figure*}

\begin{figure*}
\phantomsubfigure{fig:perf_lj_counties_A}
\phantomsubfigure{fig:perf_lj_counties_B}
\phantomsubfigure{fig:perf_lj_counties_C}
\phantomsubfigure{fig:perf_lj_counties_D}
\phantomsubfigure{fig:perf_lj_counties_E}
\phantomsubfigure{fig:perf_lj_counties_F}
\phantomsubfigure{fig:perf_lj_counties_G}
\phantomsubfigure{fig:perf_lj_counties_H}
\includegraphics[width=2\columnwidth]{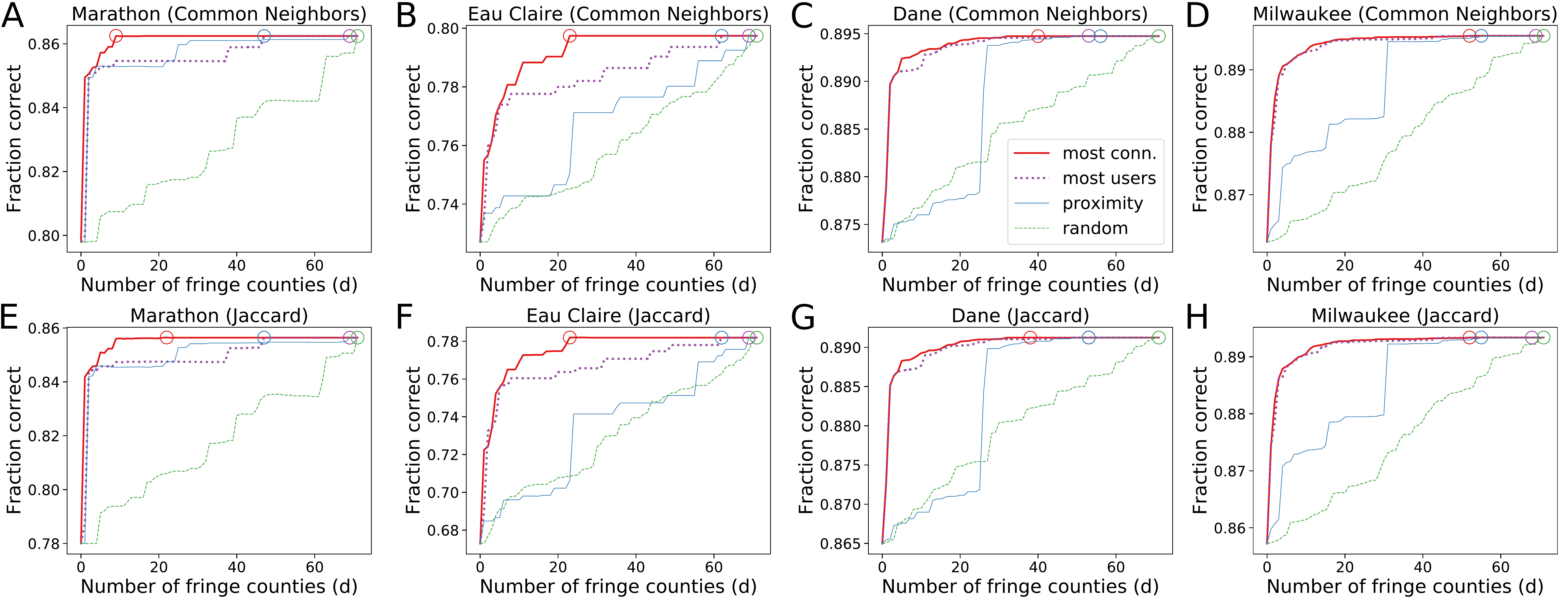}
\caption{Link prediction analogous to \cref{fig:perf_lj_states} but where the core is a Wisconsin county and the fringe is the rest of the state.
Performance is measured as a function of $d$, the number of fringe counties included for prediction.
The most connected ordering performs the best, and performance with this ordering quickly saturates.}
\label{fig:perf_lj_counties}
\end{figure*}

We now turn to links in an online social network of
bloggers from the LiveJournal community~\cite{LibenNowell-2005-geographic}.
Edges are derived from users listing friends in their profile 
(here, we consider all links as undirected). 
Users also list their geographic location, and for the purposes of this study, we have
restricted the dataset to users reporting locations in the United States and Puerto Rico.
For each user, we have both their territory of residence (one of the 50 U.S.\ states, 
Washington D.C., or Puerto Rico; henceforth simply referred to as ``state'')
as well as their county of residence, when applicable.

We construct core-fringe networks in two ways.
First, we form a core from all user residing in a given state $S$.
The core-fringe network then consists of all friendships
where at least one node is in state $S$.
We construct four such networks, using
the states Wisconsin, Texas, California, and New York.
Second, we form a core from users residing in a county and construct a core-fringe
network in the same way, but we only consider friendships amongst users in the state
containing the county. We construct four such networks,
using Marathon, Eau Claire, Dane, and Milwaukee counties (all in Wisconsin).
\Cref{tab:lj_datasets} lists summary statistics of the datasets.

Unlike the email and telecommunications networks, we do not have timestamps
on the edges. We instead construct $\edgetest$ from a random sample
of 20\% of the edges between core nodes, i.e., from $\{(u, v) \in E \;\vert\; u, v \in C\}$,
and set $\edgetrain = E - \edgetest$.
Predicting on such held out test sets is used for predicting missing links~\cite{Clauset-2008-hierarchical,Ghasemian-2018-evaluating};
here, we use it for link prediction, as is standard practice~\cite{Lu-2011-link-pred}.

\xhdr{Fringe ordering}
We again incorporate fringe nodes from a nested sequence of node sets 
$\{ \nodeparam{d} \}$, where $\nodeparam{d} \subseteq \nodeparam{d + 1}$ and $\nodeparam{0} = C$, the set of core nodes.
The nested sequence of training sets is then
$\edgetrainparam{d} = \{(u, v) \in \edgetrain \;\vert\; u, v \in \nodeparam{d} \}$.
With email and telecommunications, we considered fringe nodes
one by one to form the sequence $\{ \nodeparam{d} \}$. 
For LiveJournal,
each successive node set instead corresponds to adding all nodes in a state or a county.
For the cores constructed from users in a state (Wisconsin, Texas, New York, or California),
we form orderings $\sigma$ of the remaining states in four ways:
(i) \emph{Most connected}: $\sigma$ is the ordering of states by decreasing number of links to the core state;
(ii) \emph{Most users}: $\sigma$ is the ordering by decreasing number of users;
(iii) \emph{Proximity}: $\sigma$ is the ordering of states by closest proximity to the core state 
(measured by great-circle distance between geographic centers); and
(iv) \emph{Random}: $\sigma$ is a random ordering of the states.

Let $U_S$ be the users in state $S$.
The sequence of vertex sets is all users in the core and first $d$
states in the ordering $\sigma$. Formally, $\nodeparam{d} = C \phantom{!} \cup \phantom{!} (\cup_{t=1}^{d} U_{\sigma_t})$.
For networks whose core are users in a Wisconsin county,
we use the same orderings, except we order counties instead of states
and the fringe is only counties in Wisconsin.

\xhdr{Link prediction}
We measure the mean prediction performance over 10 random trials
as a function of the number $d$ of fringe states 
or counties included in the training data.
When states form the core, prediction performance is largely consistent (\cref{fig:perf_lj_states}).
For both score functions,
ordering by number of connections tends to perform the best,
with a rapid performance increase from approximately the first 10 states and
then saturation with a slow monotonic increase.
The prediction by states with the most
users performs nearly the same for the three largest states (Texas, New York, and California; \cref{fig:perf_lj_states_B,fig:perf_lj_states_C,fig:perf_lj_states_D}).
In Wisconsin and New York, ordering by proximity shows a steep
rise in performance for the first few states but then levels off (\cref{fig:perf_lj_states_A,fig:perf_lj_states_C,fig:perf_lj_states_E,fig:perf_lj_states_G}).
On the other hand, in California and Texas, ordering by proximity performs
roughly as well as a random ordering.

The networks where the cores are users from counties in Wisconsin have similar
characteristics to the networks where the cores are users from particular states (\cref{fig:perf_lj_counties}).
The ordering by county with the most connections performs the best.
Prediction performance quickly saturates in the two larger counties (\cref{fig:perf_lj_counties_C,fig:perf_lj_counties_D}),
and the proximity ordering can be good in some cases (\cref{fig:perf_lj_counties_A}).

\xhdr{Summary}
Usually, collecting additional data is thought to improve performance in machine learning.
Here we have seen that this is not the case in some networks with core-fringe structure.
In fact, including additional fringe information can affect link prediction performance in a number of ways.
In some cases, it is indeed true that additional fringe always helps, which was largely the case with LiveJournal (\cref{fig:perf_lj_states}).
In one email network, including any fringe data hurt performance (\cref{fig:perf_email_C}).
And yet in other cases, some intermediate amount of fringe data gave the best performance
(\cref{fig:perf_telecomms_A,fig:perf_telecomms_B}; \cref{fig:perf_enron}).
We also observed saturation in link prediction performance as we increased
the fringe size (\cref{fig:perf_lj_counties}) and that sometimes we need enough fringe before prediction
becomes better than incorporating no fringe at all (\cref{fig:perf_email_B,fig:perf_enron_C}).
While this landscape is complex, we show in the next section
how these behaviors can emerge in simple random graph models.

\section{Random graph model analysis}\label{sec:theory}
We now turn to the question of \emph{why} 
link prediction in core-fringe networks
exhibits such a wide variation in performance.
To gain insight into this question, we analyze the link prediction
problem on basic random graph models that have been
adapted to contain core-fringe structure.

Recall how our link prediction problem is evaluated: 
we are given two pairs $\{u, v\}$
and $\{w, z\}$; our algorithm predicts which of the two candidate edges is
the one that appears in the data through a score function; 
and the values of the score function (and hence the predictions) 
can change based on the inclusion of
fringe nodes. In a random graph model, we can think about using the same
score functions for the candidate edges $(u, v)$ and $(w, z)$, but now the
score functions and the existence of edges are random variables. 
As we will show,
the signal-to-noise ratio of the difference in score functions
is a key statistic to optimize in order to make the most accurate predictions,
and this can vary in different ways when including fringe nodes.

\xhdr{The signal-to-noise ratio (SNR)}
Suppose our data is generated by a random graph model and
that the indicator random variables $X, Y \in \{0,1\}$
correspond to the existence of two candidate edges $\{u,v\}$ and $\{w,z\}$, respectively,
where nodes $u$, $v$, $w$, and $z$ are distinct and chosen uniformly at random
amongst a set of core nodes.
Without loss of generality, we assume that $\prob{X} > \prob{Y}$ so that $(u, v)$
is more likely to appear.

We would like our algorithm to predict that the edge $\{u,v\}$ is the one 
that exists, since this is the more likely edge (by assumption).
However, our algorithm does not observe $X$ and $Y$ directly; instead, it sees
proxy measurements $(\proxyX, \proxyY)$, which are themselves random variables.
These proxy measurements correspond to the score function used by the algorithm;
in this section, we focus on the number of common neighbors score.
Our algorithm will (correctly) predict that edge $\{u, v\}$ is more likely if and only if $\proxyX > \proxyY$.

Furthermore, the proxy measurements are parameterized by the amount of fringe
information we have. Following our previous notation, we call
these random variables $\proxyXparam{d}$ and $\proxyYparam{d}$. These
variables represent the same measurements as $\proxyX$ and $\proxyY$ (such
as the number of common neighbors), just on a set of graphs 
parameterized by the amount of fringe information $d$.

Our goal is to optimize the amount of fringe to get the most accurate predictions.
Formally, if we let $\proxyZparam{d} \triangleq \proxyXparam{d} - \proxyYparam{d}$, this means:
\begin{align*}
\textstyle \underset{d}{\text{maximize}} \quad 
\prob{\proxyXparam{d} > \proxyYparam{d}} \iff 
\underset{d}{\text{maximize}} \quad \prob{\proxyZparam{d} > 0}.
\end{align*}
We assume that we have access to $\expec{\proxyZparam{d}}$ and $\var{\proxyZparam{d}}$
for all values of $d$. 
Our approach will be to maximize
the signal-to-noise ratio (SNR) statistic of $\proxyZparam{d}$, i.e.,
\begin{align*}
\textstyle \underset{d}{\text{maximize}} \quad \frac{\expec{\proxyZparam{d}}}{\sqrt{\var{\proxyZparam{d}}}} \triangleq \snr{\proxyZparam{d}}.
\end{align*}
We motivate this approach as follows.
Absent any additional information beyond the expectation and variance,
we must use some concentration inequality. Under the reasonable assumption
that $\expec{\proxyZparam{d}} > 0$ (which we later show holds in our models),
the proper inequality is Cantelli's:
$\prob{\proxyZparam{d} \ge 0} \ge 1 - \frac{\var{\proxyZparam{d}}}{\var{\proxyZparam{d}} + \expec{\proxyZparam{d}}^2} = \frac{\snr{\proxyZparam{d}}^2}{1 + \snr{\proxyZparam{d}}^2}$.
Thus, the lower bound on correctly choosing $X$ is monotonically increasing
in the SNR of our proxy measurement.
Using this probabilistic framework,
we can now see how some of the empirical behaviors
in \cref{sec:empirical} might arise.

\subsection{Stochastic block models}
In the stochastic block model, the nodes are partitioned into 
$K$ {\em blocks}, and for parameters $P_{i,j}$ (with $1 \leq i, j \leq K$), 
each node in block $i$
is connected to each node in block $j$ independently with 
probability $P_{i,j}$.
(Since our graphs are undirected, $P_{i, j} = P_{j, i}$.)
In our core-fringe model here, $K=4$, the core corresponds to
blocks $1$ and $2$, and the fringe corresponds to blocks $3$ and $4$.
This model turns out to be flexible enough to demonstrate a wide range of
behaviors observed in \cref{sec:empirical}.
We use the following notation for block probabilities:
\begin{equation}\label{eq:cf_sbm}
P = \begin{bmatrix}
p & q & r & s \\
q & p & s & r \\
r & s & 0 & 0\\
s & r & 0 & 0 \\
\end{bmatrix}.
\end{equation}
Our assumptions on the probabilities are that $q < p$ and $s \le r$.
We also assume that the first two blocks each contain $n_c$ nodes
and that the last two blocks each contain $n_f$ nodes.

We further assume we are given samples of 
four distinct nodes $u$, $v$, $w$, and $z$ chosen
uniformly at random from the core blocks such that $u$, $v$, and $w$ are in block 1
and $z$ is in block 2. Following our notation above, let
$X$ be the random variable that $(u, v)$ is an edge and $Y$ be the random variable
that $(w, z)$ is an edge ($\prob{X} > \prob{Y}$ since $p > q$). 
Our proxy measurements $\proxyX$ and $\proxyY$ are
the number of common neighbors of candidate edges $\{u, v\}$ and $\{w, z\}$.
Our algorithm will correctly predict that $(u, v)$ is more likely if $\proxyX > \proxyY$.

Our proxy measurements are parameterized by the amount of fringe
information that they incorporate. Here, we arbitrarily order the nodes
in the two equi-sized fringe blocks and say that the random variable $\proxyXparam{d}$
is the number of common neighbors between nodes $u$ and $v$ when including
the first $d$ nodes in both fringe blocks.
Similarly, the random variable $\proxyYparam{d}$ is 
the number of common neighbors of nodes $w$ and $z$.
By independence of the edge probabilities, some straightforward calculations show that
\begin{align*}
\expec{\proxyXparam{d}} &= 2(n_c - 1)p^2 + dr^2 + ds^2,\; \expec{\proxyYparam{d}} = 2(n_c - 1)pq + 2drs \\
\var{\proxyXparam{d}}  &= 2(n_c - 1)p^2(1 - p^2) + dr^2(1 - r^2) + ds^2(1 - s^2) \\
\var{\proxyYparam{d}} &= 2(n_c - 1)pq(1 - pq) + 2drs(1 - rs).
\end{align*}
With no fringe information, it is immediate that
the SNR is positive, i.e., $\snr{\proxyZparam{0}} > 0$:
$\expec{\proxyXparam{0} - \proxyYparam{0}} = 2(n_c - 1)p[p - q] > 0$ as $p > q$.
If the two fringe blocks connect to the two core blocks with equal probability,
then the SNR with no fringe is optimal.
\begin{lemma}[No-fringe optimality]\label{lem:SBM_no_fringe}
If $r = s$ in the core-fringe SBM, then $\snr{\proxyZparam{d}}$ decreases monotonically  in $d$.
\end{lemma}
\begin{proof}
When $r = s$, by independence of $\proxyXparam{d}$ and $\proxyYparam{d}$,
\begin{align*}
\expec{\proxyZparam{0}} = \expec{\proxyXparam{d} - \proxyYparam{d}} 
&= 2(n_c - 1)p^2 + dr^2 + ds^2 - 2(n_c - 1)pq - 2drs \\
&= 2(n_c - 1)p^2 - 2(n_c - 1)pq = \expec{\proxyXparam{0} - \proxyYparam{0}},
\end{align*}
and
$\var{\proxyZparam{d}} 
= \var{\proxyXparam{d}} + \var{\proxyYparam{d}} 
> \var{\proxyXparam{0}} + \var{\proxyYparam{0}} = \var{\proxyZparam{0}}$.
\end{proof}
This result is intuitive. If the two fringe blocks connect to the core nodes
with equal probability, then the node pairs $(u, v)$ and $(w, z)$ receive
additional extra common neighbors according to exactly
the same distributions. Thus, these
fringe nodes provide noise but no signal.
We confirm this result numerically with parameter
settings $p = 0.5$, $q = 0.3$, $s = r = 0.2$, and $n_c = 10$ (\cref{fig:SBM_A}).
Indeed, the SNR monotonically decreases as we include more fringe.

\begin{figure}[t]
\phantomsubfigure{fig:SBM_A}
\phantomsubfigure{fig:SBM_B}
\includegraphics[width=\columnwidth]{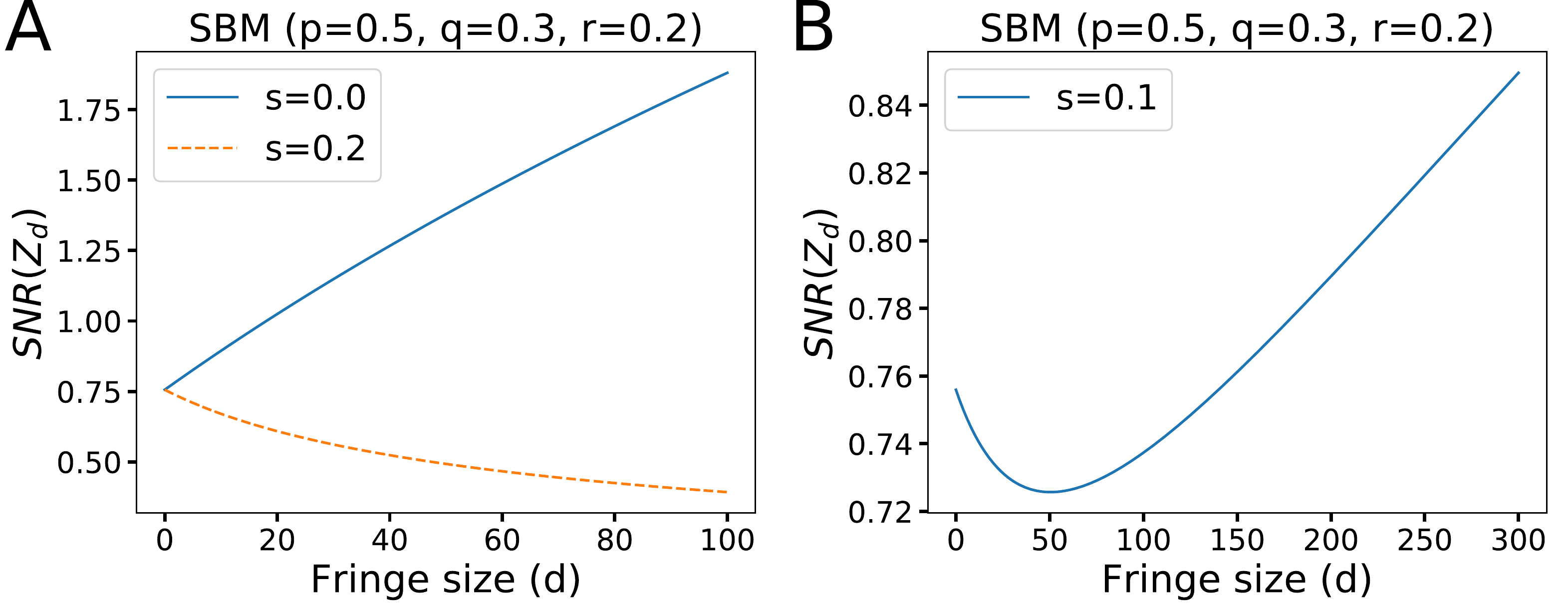}
\caption{SNR for the difference in 
the common neighbors in our stochastic block model
of core-fringe structure with $n_c = 10$. (A) When the fringe blocks have
equal probability of connecting to the two core blocks ($r = s$),
the SNR decreases monotonically with fringe size by \cref{lem:SBM_no_fringe}.
When fringe blocks only connect to one of the core blocks ($s = 0.0$),
the SNR increases monotonically with fringe size \cref{lem:SBM_all_fringe}.
(B) For an intermediate parameter setting, including the fringe hurts the SNR 
until enough fringe is included, at which point the SNR increases monotonically
(\cref{lem:SBM_enough_fringe}).
}
\label{fig:SBM}
\end{figure}

In \cref{sec:empirical}, we saw cases where including any additional
fringe information always helped.
The following lemma shows that we can set parameters in our 
core-fringe SBM such that including any additional fringe
information always increases the SNR.
\begin{lemma}[All-fringe optimality]\label{lem:SBM_all_fringe}
Let $r > 0$, $s = 0$ and $4(n_c - 1)(p^2 - p^4) > 1$ in the core-fringe
SBM. Then $\snr{\proxyZparam{d}}$ monotonically increases in $d$
and $\lim_{d \to \infty}\snr{\proxyZparam{d}} = \infty$.
\end{lemma}
\begin{proof}
We have that
\[
\textstyle \snr{\proxyZparam{d}} = 
\frac{
2(n_c - 1)p(p - q) + dr^2
}{
\sqrt{2(n_c - 1)(p^2(1 - p^2) + pq(1 - pq)) + dr^2(1 - r^2)}
}.
\]
We can treat this function as continuous in $d$. Then
\[
\textstyle \lim_{d \to \infty}\snr{\proxyZparam{d}} = \lim_{d \to \infty}\frac{dr^2}{\sqrt{dr^2(1 - r^2)}} = \infty.
\]
Similarly, we can compute the derivative with respect to $d$:
\[
\textstyle \frac{\partial}{\partial d}\snr{\proxyZparam{d}} = 
\frac{
r^2\sqrt{\var{\proxyZparam{d}}} - \frac{1}{2}r^2(1 - r^2) / \sqrt{\var{\proxyZparam{d}}}
}{
\var{\proxyZparam{d}}
}
\]
The derivative is positive provided that
$r^2\var{\proxyZparam{d}} > \frac{1}{2}r^2(1 - r^2)$,
which is true if $\var{\proxyZparam{0}} > \frac{1}{2}$
since $\var{\proxyZparam{0}}$ is monotonically increasing in $d$. 
We have that
\[
\var{\proxyZparam{0}} = 2(n_c - 1)(p^2(1 - p^2) + pq(1 - pq)) > 2(n_c - 1)(p^2 - p^4).
\]
Thus, the result holds provided $4(n_c - 1)(p^2 - p^4) > 1$.
\end{proof}

The result is again intuitive.
By setting $s = 0$, the pair of nodes in different blocks ($w$ and $z$) get no additional
common neighbors from the fringe, whereas the pair of nodes in the
same block ($u$ and $v$) get additional common neighbors. This should only help
our prediction performance, which is why the SNR monotonically
increases. We confirm this result numerically (\cref{fig:SBM_A}), where
we use the same parameters as the experiment described above.
We can also check that the conditions of \cref{lem:SBM_all_fringe} are met:
$n_c = 10$ and $p = 0.5$, so $4(n_c - 1)(p^2 - p^4) = 6.75 > 1$.

The SBM can also exhibit additional behaviors that we observed
in \cref{sec:empirical}. For example, with the email-Enron-4 dataset,
prediction performance initially decreased as we included more fringe
and then began to increase (\cref{fig:perf_enron_D,fig:perf_enron_H}).
The following lemma says that the core-fringe SBM can capture this behavior.
\begin{lemma}[Enough-fringe optimality]\label{lem:SBM_enough_fringe}
Let $p$, $q$, and $r$ be given. Then there exists a value of $s$
in the core-fringe SBM
such that $\snr{\proxyZparam{d}}$ initially decreases and then
increases without bound.
\end{lemma}
\begin{proof}
To simplify notation, consider the following constants:
$\alpha = \expec{\proxyZparam{0}} = 2(n_c - 1)(p^2 - pq)$,
$\beta = (r - s)^2$,
$\gamma = \var{\proxyZparam{0}} = 2(n_c - 1)[p^2(1 - p^2) + pq(1 - pq)]$, and
$\delta = r^2(1 - r^2) + s^2(1 - s^2) + 2rs(1 - rs)$.
With this notation, $\snr{\proxyZparam{d}} = (\alpha + \beta d) / \sqrt{\gamma + \delta d}$.
Treating this as a continuous function in $d$, the derivative is:
$\frac{\partial}{\partial d} \snr{\proxyZparam{d}} 
= (-\alpha\delta + 2\beta\gamma + \beta\delta d) / (2(\gamma + \delta d)^{3/2})$.
For any $s$, we can choose a sufficiently large $D$ such that
the derivative is positive when $d > D$, meaning that $\snr{\proxyZparam{d}}$ is increasing.
Furthermore, $\snr{\proxyZparam{d}}$ grows as $O(\sqrt{d})$.
It is easy to check that the derivative also has at most one root, 
$d_0 = \alpha/\beta - 2\gamma/\delta$.
We claim that $d_0$ can be made as large as desired.
By setting $s$ sufficiently close to $r$, $\beta$ approaches $0$,
while $\delta$ is bounded away from $0$. 
The remaining terms are positive constants.
Finally, when $d = 0$, the value of the derivative is
$(-\alpha\delta + 2\beta\gamma) / (2\gamma^{3/2})$.
Again, we can make $s$ sufficiently close to $r$ so that $\beta$ approaches $0$
and the derivative is negative at $d = 0$.
Therefore, there exists an $s$ such that 
its derivative has one root $d_0 \ge 1$,
$\snr{\proxyZparam{d}}$ decreases for small enough $d$ and
eventually increases without bound.
\end{proof}
By setting $n_c = 10$, $p = 0.5$, $q = 0.3$, $r = 0.2$, and $s = 0.1$,
we see the behavior described by \cref{lem:SBM_enough_fringe}---the
SNR initially decreases with additional fringe but then
increases monotonically (\cref{fig:SBM_B}).

By extending the SBM to include a third fringe block,
we can also have a case where an intermediate amount of core is optimal.
We argue informally as follows.
We begin with a setup as in \cref{lem:SBM_all_fringe}, where $s = 0$.
Including all of the fringe available in these blocks is optimal.
We then add a third fringe block that connects with equal probability to
the two core blocks.
By the arguments in \cref{lem:SBM_no_fringe}, this only hurts the SNR.
Thus, it is optimal to include two of the three fringe blocks, which is an
intermediate amount of fringe.

\subsection{Small-world lattice models}

In the one-dimensional small-world lattice model~\cite{Kleinberg-2006-ICM},
there is a node for each integer in $\mathbb{Z}$ and a parameter $\alpha \ge 0$.
The probability that edge $(i, j)$ exists is
\begin{align}\label{eq:1d_lattice_prob}
\textstyle \prob{(i, j) \in E} = \frac{1}{\vert j - i \vert^{\alpha}}.
\end{align}
We start with a core of size $2c + 1$, centered around 0:
$\nodeparam{0} = C = \{-c, \ldots, c\}$.
We then sample two nodes $v$ and $w$ such that
\begin{equation}\label{eq:sample}
u = -c < v < w < c = z \text{ and } 2 \le v - u < z - w.
\end{equation}
In our language at the beginning of \cref{sec:theory}, $X$ is still
the random variable that edge $(u, v)$ exists and $Y$ is the random variable that
edge $(w, z)$ exists. By our assumptions and \cref{eq:1d_lattice_prob}, 
we know that $\prob{X} > \prob{Y}$. However, we will again assume that we are
only given access to the number of common neighbors through
the proxy random variables $\proxyX$ and $\proxyY$.

Our parameterization of the proxy measurements are a distance $d$ that we examine
beyond the core. Specifically, the nested sequence of vertex sets that incorporate fringe
information is given by $\nodeparam{d} = \{-(c + d), \ldots, c + d\}$,
and our proxy measurements are
\begin{align*}
  \proxyXparam{d} &= \lvert \{ s \in \nodeparam{d} \given (u, s) \text{ and } (v, s) \text{ are edges} \} \rvert \\
  \proxyYparam{d} &= \lvert \{ s \in \nodeparam{d} \given (w, s) \text{ and } (z, s) \text{ are edges} \} \rvert. 
\end{align*}
We will analyze the random variable $\proxyZparam{d} = \proxyXparam{d} - \proxyYparam{d}$.
We correctly predict that $(u, v)$ is more likely than $(w, z)$ to exist if $\proxyZparam{d} > 0$.
As argued above, our goal is to find a $d$ that maximizes $\snr{\proxyZparam{d}}$.

We focus our analysis on the case of $\alpha = 1$ in \cref{eq:1d_lattice_prob}.
Let $A_s$ be the indicator random variable that node $s$ is
a common neighbor of nodes $u$ and $v$, 
for $s \in \mathbb{Z} \backslash \{u, v\}$,
and let $B_s$ be the indicator random variable that node $s$ is a common
neighbor of $w$ and $z$, for $s \in \mathbb{Z} \backslash \{w, z\}$.
Since $u = -c$ and $z = c$, our proxy measurements are
\begin{align}
\proxyXparam{d} &= \textstyle \sum_{s = -(c+d)}^{-(c+1)}A_s + \sum_{s = -c+1}^{v-1}A_s + \sum_{s = v+1}^{c+d}A_s \label{eq:1DL_proxyX} \\
\proxyYparam{d} &= \textstyle \sum_{s = -(c+d)}^{w-1}B_s + \sum_{s = w+1}^{c-1}B_s + \sum_{s = c+1}^{c+d}B_s \label{eq:1DL_proxyX}
\end{align}
Define the independent indicator random variables 
$I_{s,r}$ and $J_{s,r}$ where
$\prob{I_{s,r} = 1} = 1 / (s(s + r))$ and
$\prob{J_{s,r} = 1} = 1 / (s(r - s))$.
Now we can re-write the expressions for $\proxyXparam{d}$ and $\proxyYparam{d}$ as follows:
\begin{align}
\proxyXparam{d} &= \textstyle 
    \sum_{s = 1}^{d}I_{s,v-u} 
+  \sum_{s = 1}^{c + d - v}I_{s,v-u}
+  \sum_{s = 1}^{v - u - 1}J_{s,v-u} \label{eq:indicators_X}\\
\proxyYparam{d} &= \textstyle 
    \sum_{s = 1}^{d}I_{s,z-w}
+  \sum_{s = 1}^{w - c - d}I_{s,z-w}
+  \sum_{s = 1}^{z - w - 1}J_{s,z-w}.
\end{align}
The expectations are given by
\begin{align*}
\expec{\proxyXparam{d}} &= \textstyle 
    \sum_{s = 1}^{d}\frac{1}{s(s+v-u)}
+  \sum_{s = 1}^{c + d - v}\frac{1}{s(s+v-u)}
+  \sum_{s = 1}^{v - u - 1}\frac{1}{s(v - u - s)} \\
\expec{\proxyYparam{d}} &= \textstyle
    \sum_{s = 1}^{d}\frac{1}{s(s+z-w)}
+  \sum_{s = 1}^{w - c - d}\frac{1}{s(s+z-w)}
+  \sum_{s = 1}^{z - w - 1}\frac{1}{s(z-w-s)}.
\end{align*}

With these expressions, we can now analyze how $\proxyZparam{d}$
behaves as we vary $d$.
The following lemma establishes that $\proxyZparam{d}$ converges to a positive
value. Later, we use this to show that the SNR also
converges to a positive value.
\begin{lemma}\label{lem:1DL_pos}
$\lim_{d \to \infty} \expec{\proxyZparam{d}} = Z^* > 0$.
\end{lemma}
\begin{proof}
Let $a = v - u$. Then
$\lim_{d \to \infty} \expec{\proxyXparam{d}} = 
2\sum_{s=1}^{\infty} \frac{1}{s(s + a)} + \sum_{s = 1}^{a - 1}\frac{1}{s(a-s)} 
=
2(\psi(a + 1) + \psi(a) + 2\gamma) / a = 2(2\psi(a) + 1 / a + 2\gamma) / a$,
where $\psi(\cdot)$ is the digamma function.
Similarly, if $b = z - w$, then
\[
\lim_{d \to \infty} \expec{\proxyYparam{d}} = 2(2\psi(b) + 1 / b + 2\gamma) / b.
\]
Thus, 
$Z^* = \lim_{d \to \infty}\expec{\proxyXparam{d}} - \expec{\proxyYparam{d}}$
exists, and $Z^* > 0$ if and only if
\begin{align}
b(\psi(a) + 1 / (2a) + \gamma) - a(\psi(b) + 1 / (2b) + \gamma) > 0 \label{eq:gamma_ineq}
\end{align}
Recall that by \cref{eq:sample}, $2 \le a < b$.
Numerically, \cref{eq:gamma_ineq} holds for $(a, b) = (2, 3)$.
Since the left-hand-side monotonically increases in $b$, this inequality holds for $a = 2$.

Now assume $b > a \ge 3$.
The Puiseux series expansion of $\psi$ at $\infty$ gives
$\psi(x) + 1 / (2x) \in \log(x) \pm \frac{1}{12(x^2 - 1)}$.
Thus, it is sufficient to show that
$b(\log(a) - 1 / 96) + (b - a)\gamma > a(\log(b) + 1 / 180)$,
or that $0.99b\log(a) + \gamma > 1.01a\log(b)$, which holds
for $b > a \ge 3$.
\end{proof}

The next theorem shows that the signal-to-noise ratio converges
to a positive value. Thus, by measuring enough fringe, our proxy
measurements are at least providing the correct direction of information.
However, the SNR converges, so at some point, our information saturates.
\begin{theorem}[SNR saturation]\label{thm:SNR_conv}
$\lim_{d \to \infty} \snr{\proxyZparam{d}} = S^* > 0$.
\end{theorem}
\begin{proof}
By \cref{lem:1DL_pos}, $\expec{\proxyZparam{d}}$ converges to a positive value.
Thus, it is sufficient to show that $\var{\proxyZparam{d}}$ converges.
Following \cref{eq:indicators_X},
\[
\textstyle
\var{\sum_{s=1}^{\infty} I_{s,v-u}} 
= \textstyle \sum_{s=1}^{\infty}\var{I_{s,v-u}} 
= \textstyle \sum_{s=1}^{\infty}\frac{1}{s(s + v-u)} - \frac{1}{(s(s + v-u))^{2}}
\]
by independence, and converges.
\end{proof}

\noindent The random variable 
$W_{d + 1, J} \triangleq \proxyZparam{d + J} - \proxyZparam{d} = \sum_{k=d + 1}^{k=d+1+J}I_{k,v-u} - I_{k,z-w}$
will be useful for our subsequent analysis. This is the additional measurement available to us 
if we measured at distance $d + J$ instead of $d$.
The next lemma says that, as we increase $d$,
the expectation of $W_{d,J}$ goes to zero faster than its variance.
\begin{lemma}\label{lem:ev_zero}
For any $J$, $\lim_{d \to \infty} \expec{W_{d,J}} / \var{W_{d,J}} = 0$.
\end{lemma}
\begin{proof}
Let $a = v - u$ and $b = z - w$. By independence,
\begin{align*}
\expec{W_{d,J}} 
&= \textstyle \sum_{k=d}^{k=d+J}\frac{1}{k(k + a)} - \frac{1}{k(k + b)} \\
&= \textstyle \sum_{k=d}^{k=d+J}\frac{b-a}{k(k+a)(k+b)} = \textstyle O(\sum_{k=d}^{k=d+J}1 / k^3).
\end{align*}
For large enough $d$,
$\var{W_{d,J}} = \textstyle \sum_{k=d}^{k=d+J}\frac{1}{k(k + a)} + \frac{1}{k(k + b)} - \frac{1}{(k(k + b))^2} - \frac{1}{(k(k + a))^2}$,
which is $O(\sum_{k=d}^{k=d+J}1 / k^2)$.
\end{proof}
\noindent The next theorem now shows that in the one-dimensional lattice model,
the signal-to-noise ratio eventually begins to decrease.
This means that at some point, the noise overwhelms the signal.
Thus, it will never be best to gather as much fringe as possible.
\begin{theorem}\label{thm:decrease}
There exists a $D$ for which $\snr{\proxyZparam{D}} > \snr{\proxyZparam{D + j}}$ for any $j > 0$.
\end{theorem}
\begin{proof}
We have that $\snr{\proxyZparam{d}} > \snr{\proxyZparam{d + j}}$ if and only if
\begin{align*}
\textstyle \frac{\expec{\proxyZparam{d}}}{\sqrt{\var{\proxyZparam{d}}}} > \frac{\expec{\proxyZparam{d} + W_{d+1,j}}}{\sqrt{\var{\proxyZparam{d} + W_{d+1,j}}}} 
&\iff \textstyle \frac{\expec{\proxyZparam{d}}^2}{\var{\proxyZparam{d}}} > \frac{2\expec{\proxyZparam{d}}\expec{W_{d+1,j}} + \expec{W_{d+1,j}}^2}{\var{W_{d+1,j}}}.
\end{align*}
The second inequality above comes from 
squaring both sides of the first inequality;
both terms in the first inequality are positive for large enough $d$ by \cref{lem:1DL_pos}, 
so we can keep the direction of the inequality.
By \cref{thm:SNR_conv}, the left-hand-side of the inequality converges to a positive constant.
We claim that the right-hand-side converges to $0$.

\begin{figure}[t]
\includegraphics[width=\columnwidth]{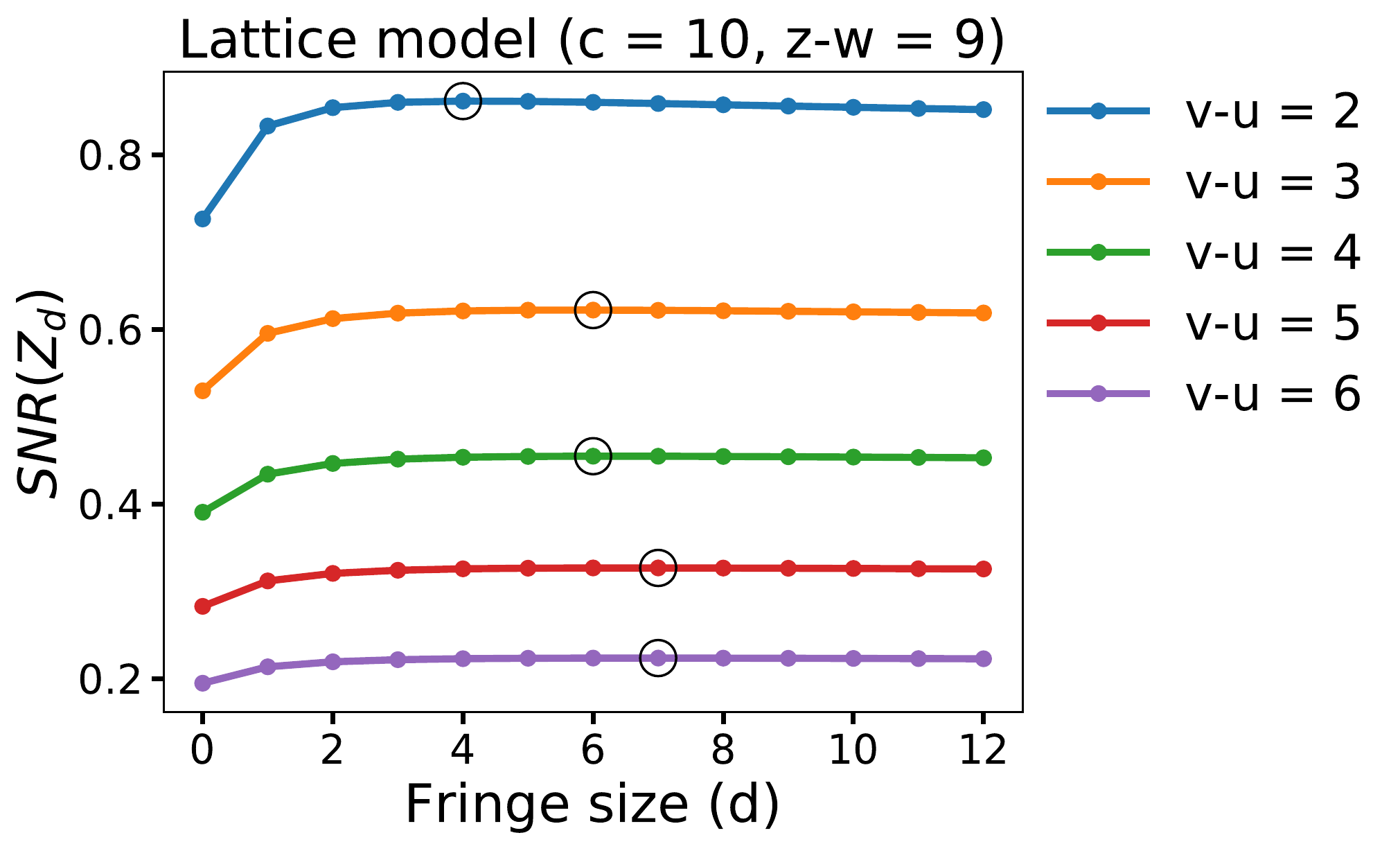}
\caption{Saturation and interior optima of the SNR
the one-dimensional small-world lattice model.
Nodes $v$ and $w$ are sampled from $\{u = -c, \ldots, c = z\}$
with $2 \le v - u < z - w$. The random variable $\proxyZparam{d}$
is the difference in the number of common neighbors of $\{u,v\}$
and $\{w,z\}$ on the node set $\{-(c + d), \ldots, c + d\}$.
Since $\snr{\proxyZparam{1}} > \snr{\proxyZparam{0}}$, 
an intermediate amount of fringe produces the optimal SNR
(\cref{cor:intermediate}); these optima are circled in black.
The SNR converges by \Cref{thm:SNR_conv}, and we indeed see the
consequent saturation.}
\label{fig:lattice_results}
\end{figure}

By \cref{lem:1DL_pos}, $\expec{\proxyZparam{d}}$ converges, so it must be bounded by a positive constant constant.
Furthermore, since $\proxyZparam{d + j} = \proxyZparam{d} + W_{d+1,j}$, we have that $\expec{W_{d+1,j}} \to 0$.
Combining these results,
$2\expec{\proxyZparam{d}}\expec{W_{d+1,j}} + \expec{W_{d+1,j}}^2 = O(\expec{W_{d+1,j}})$,
and we have that $\expec{W_{d+1,j}} / \var{W_{d+1,j}} \to 0$ by \cref{lem:ev_zero}.
\end{proof}
A consequence of this theorem is that if the SNR initially increases,
then an intermediate amount of fringe is optimal.
The reason is that 
the SNR initially increases but at some point begins to decrease monotonically (\cref{thm:decrease})
before converging to a positive value (\cref{thm:SNR_conv}).
We formalize this as follows.
\begin{corollary}[Intermediate-fringe optimality]\label{cor:intermediate}
  If\; $\snr{\proxyZparam{0}} < \snr{\proxyZparam{1}}$,
  then $d^* = \arg\max_{d}\snr{\proxyZparam{d}}$ satisfies $0 < d^* < \infty$.
\end{corollary}

Numerically, $\snr{\proxyZparam{0}} < \snr{\proxyZparam{1}}$ in several cases (\cref{fig:lattice_results}).
In this experiment, we fix $c = 10$ and $w = 1$ (so $z - w = 9$)
and vary the amount of fringe from $0$ to $12$ nodes on either end of the core.
We also vary $v$ so that $v - u \in \{2, 3, 4, 5, 6\}$.
We observe two phenomena consistent with our theory.
First, by \cref{cor:intermediate}, an intermediate amount of fringe
information should be optimal; indeed, this is the case.
Second, by \cref{thm:SNR_conv}, the SNR
converges, indicating that saturation should kick in at some finite fringe size. This
is true in our experiments, where saturation occurs after around  $d = 8$.

\section{Discussion}
Link prediction is a cornerstone problem in network
science~\cite{LibenNowell-2007-link-pred,Lu-2011-link-pred},
and the models for prediction include those that are
mechanistic~\cite{Barabasi-2002-evolution},
statistical~\cite{Clauset-2008-hierarchical},
or implicitly captured by a principled heuristic~\cite{Adamic-2003-friends,Backstrom-2011-supervised}.
The major difference in our work is that we explicitly study the consequences
of a common dataset collection process that results in core-fringe structure.
Most related to our analysis of random graph models are theoretical justifications
of principled heuristics such as the number of common neighbors in
latent space models~\cite{Sarkar-2011-justifications}
and in general stochastic block models~\cite{Sarkar-2015-consistency}.

The core-fringe structure that we study can be interpreted as an extreme case of
core-periphery structure in complex
networks~\cite{Borgatti-2000-CP,Holme-2005-CP,Rombach-2017-CP,Zhang-2015-SBM-CP}.
In more classical social and economic network analysis, core-periphery structure
is a consequence of differential
status~\cite{Doreian-1985-structural,Laumann-1976-collective}. In this paper,
the structure emerges from data collection mechanisms, which raises
new research questions of the kind that we have addressed.
However, our results hint that periphery nodes could also be noisy sources of information
and possibly warrant omission in standard link prediction.
Our fringe measurements can also be viewed as adding noisy training
data, which is related to training data augmentation methods~\cite{Ratner-2016-data,Ratner-2017-Snorkel}.

Conventional machine learning wisdom says that more data generally
helps make better predictions. We showed that this is far from true in the common
problem of network link prediction, where additional data comes from
observing how some core set of nodes interacts with the rest of the world, inducing
core-fringe structure. Our empirical results show that the
inclusion of additional fringe information leads to substantial variability in
prediction performance with common link prediction heuristics.
We observed cases where fringe information is (i) always harmful, (ii) always
beneficial, (iii) beneficial only \emph{up to} a certain amount of collection,
and (iv) beneficial only with \emph{enough} collection.

At first glance, this variability seems difficult to characterize. However, we
showed that these behaviors arise in some simple graph models---namely, the
stochastic block model and the one-dimensional small-world lattice model---by interpreting
the benefit of the fringe information as
changing the signal-to-noise ratio in our prediction problem. Our datasets are certainly more
complex than these models, but our analysis suggests that variability
in prediction performance when incorporating fringe data is much more plausible
than one might initially suspect.
Even when fringe data is available in network analysis, we must be careful how
we incorporate this data into the prediction models we build.

\subsection*{Data and software}
The call-Reality, text-Reality, email-Enron, and email-W3C datasets,
along with software for reproducing the results in this paper,
are available at \url{https://github.com/arbenson/cflp}.

\subsection*{Acknowledgements}
We thank David Liben-Nowell for access to
the LiveJournal data and Jure Leskovec for access to the email-Eu data.
This research was supported by
NSF Awards DMS-1830274, CCF-1740822, and SES-1741441;
a Simons Investigator Award;
and ARO Award W911NF-19-1-0057.

\bibliographystyle{ACM-Reference-Format}
\bibliography{refs}


\begin{thebibliography}{37}


\ifx \showCODEN    \undefined \def \showCODEN     #1{\unskip}     \fi
\ifx \showDOI      \undefined \def \showDOI       #1{#1}\fi
\ifx \showISBNx    \undefined \def \showISBNx     #1{\unskip}     \fi
\ifx \showISBNxiii \undefined \def \showISBNxiii  #1{\unskip}     \fi
\ifx \showISSN     \undefined \def \showISSN      #1{\unskip}     \fi
\ifx \showLCCN     \undefined \def \showLCCN      #1{\unskip}     \fi
\ifx \shownote     \undefined \def \shownote      #1{#1}          \fi
\ifx \showarticletitle \undefined \def \showarticletitle #1{#1}   \fi
\ifx \showURL      \undefined \def \showURL       {\relax}        \fi
\providecommand\bibfield[2]{#2}
\providecommand\bibinfo[2]{#2}
\providecommand\natexlab[1]{#1}
\providecommand\showeprint[2][]{arXiv:#2}

\bibitem[\protect\citeauthoryear{Abbe}{Abbe}{2018}]%
        {Abbe-2018-community}
\bibfield{author}{\bibinfo{person}{Emmanuel Abbe}.}
  \bibinfo{year}{2018}\natexlab{}.
\newblock \showarticletitle{Community Detection and Stochastic Block Models:
  Recent Developments}.
\newblock \bibinfo{journal}{\emph{Journal of Machine Learning Research}}
  \bibinfo{volume}{18}, \bibinfo{number}{177} (\bibinfo{year}{2018}),
  \bibinfo{pages}{1--86}.
\newblock
\urldef\tempurl%
\url{http://jmlr.org/papers/v18/16-480.html}
\showURL{%
\tempurl}


\bibitem[\protect\citeauthoryear{Abbe, Bandeira, and Hall}{Abbe
  et~al\mbox{.}}{2016}]%
        {Abbe-2016-exact}
\bibfield{author}{\bibinfo{person}{Emmanuel Abbe}, \bibinfo{person}{Afonso~S.
  Bandeira}, {and} \bibinfo{person}{Georgina Hall}.}
  \bibinfo{year}{2016}\natexlab{}.
\newblock \showarticletitle{Exact Recovery in the Stochastic Block Model}.
\newblock \bibinfo{journal}{\emph{{IEEE} Transactions on Information Theory}}
  \bibinfo{volume}{62}, \bibinfo{number}{1} (\bibinfo{year}{2016}),
  \bibinfo{pages}{471--487}.
\newblock
\urldef\tempurl%
\url{https://doi.org/10.1109/tit.2015.2490670}
\showDOI{\tempurl}


\bibitem[\protect\citeauthoryear{Adamic and Adar}{Adamic and Adar}{2003}]%
        {Adamic-2003-friends}
\bibfield{author}{\bibinfo{person}{Lada~A Adamic} {and} \bibinfo{person}{Eytan
  Adar}.} \bibinfo{year}{2003}\natexlab{}.
\newblock \showarticletitle{Friends and neighbors on the web}.
\newblock \bibinfo{journal}{\emph{Social networks}} \bibinfo{volume}{25},
  \bibinfo{number}{3} (\bibinfo{year}{2003}), \bibinfo{pages}{211--230}.
\newblock


\bibitem[\protect\citeauthoryear{Backstrom and Leskovec}{Backstrom and
  Leskovec}{2011}]%
        {Backstrom-2011-supervised}
\bibfield{author}{\bibinfo{person}{Lars Backstrom} {and} \bibinfo{person}{Jure
  Leskovec}.} \bibinfo{year}{2011}\natexlab{}.
\newblock \showarticletitle{Supervised Random Walks: Predicting and
  Recommending Links in Social Networks}. In
  \bibinfo{booktitle}{\emph{Proceedings of the Fourth ACM International
  Conference on Web Search and Data Mining}}. \bibinfo{publisher}{ACM},
  \bibinfo{pages}{635--644}.
\newblock
\urldef\tempurl%
\url{https://doi.org/10.1145/1935826.1935914}
\showDOI{\tempurl}


\bibitem[\protect\citeauthoryear{Barab{\'{a}}si, Jeong, N{\'{e}}da, Ravasz,
  Schubert, and Vicsek}{Barab{\'{a}}si et~al\mbox{.}}{2002}]%
        {Barabasi-2002-evolution}
\bibfield{author}{\bibinfo{person}{A.L Barab{\'{a}}si}, \bibinfo{person}{H
  Jeong}, \bibinfo{person}{Z N{\'{e}}da}, \bibinfo{person}{E Ravasz},
  \bibinfo{person}{A Schubert}, {and} \bibinfo{person}{T Vicsek}.}
  \bibinfo{year}{2002}\natexlab{}.
\newblock \showarticletitle{Evolution of the social network of scientific
  collaborations}.
\newblock \bibinfo{journal}{\emph{Physica A: Statistical Mechanics and its
  Applications}} \bibinfo{volume}{311}, \bibinfo{number}{3-4}
  (\bibinfo{year}{2002}), \bibinfo{pages}{590--614}.
\newblock
\urldef\tempurl%
\url{https://doi.org/10.1016/s0378-4371(02)00736-7}
\showDOI{\tempurl}


\bibitem[\protect\citeauthoryear{Benson and Kleinberg}{Benson and
  Kleinberg}{2018}]%
        {Benson-2018-found-graph-data}
\bibfield{author}{\bibinfo{person}{Austin~R. Benson} {and} \bibinfo{person}{Jon
  Kleinberg}.} \bibinfo{year}{2018}\natexlab{}.
\newblock \showarticletitle{Found Graph Data and Planted Vertex Covers}. In
  \bibinfo{booktitle}{\emph{Advances in Neural Information Processing
  Systems}}.
\newblock


\bibitem[\protect\citeauthoryear{Borgatti and Everett}{Borgatti and
  Everett}{2000}]%
        {Borgatti-2000-CP}
\bibfield{author}{\bibinfo{person}{Stephen~P Borgatti} {and}
  \bibinfo{person}{Martin~G Everett}.} \bibinfo{year}{2000}\natexlab{}.
\newblock \showarticletitle{Models of core/periphery structures}.
\newblock \bibinfo{journal}{\emph{Social Networks}} \bibinfo{volume}{21},
  \bibinfo{number}{4} (\bibinfo{year}{2000}), \bibinfo{pages}{375--395}.
\newblock
\urldef\tempurl%
\url{https://doi.org/10.1016/s0378-8733(99)00019-2}
\showDOI{\tempurl}


\bibitem[\protect\citeauthoryear{Clauset, Moore, and Newman}{Clauset
  et~al\mbox{.}}{2008}]%
        {Clauset-2008-hierarchical}
\bibfield{author}{\bibinfo{person}{Aaron Clauset}, \bibinfo{person}{Cristopher
  Moore}, {and} \bibinfo{person}{M.~E.~J. Newman}.}
  \bibinfo{year}{2008}\natexlab{}.
\newblock \showarticletitle{Hierarchical structure and the prediction of
  missing links in networks}.
\newblock \bibinfo{journal}{\emph{Nature}} \bibinfo{volume}{453},
  \bibinfo{number}{7191} (\bibinfo{year}{2008}).
\newblock


\bibitem[\protect\citeauthoryear{Craswell, de~Vries, and Soboroff}{Craswell
  et~al\mbox{.}}{2005}]%
        {Craswell-2005-TREC}
\bibfield{author}{\bibinfo{person}{Nick Craswell}, \bibinfo{person}{Arjen~P de
  Vries}, {and} \bibinfo{person}{Ian Soboroff}.}
  \bibinfo{year}{2005}\natexlab{}.
\newblock \showarticletitle{Overview of the TREC 2005 Enterprise Track}. In
  \bibinfo{booktitle}{\emph{TREC}}, Vol.~\bibinfo{volume}{5}.
  \bibinfo{pages}{199--205}.
\newblock


\bibitem[\protect\citeauthoryear{Decelle, Krzakala, Moore, and
  Zdeborov{\'{a}}}{Decelle et~al\mbox{.}}{2011}]%
        {Decelle-2011-SBM}
\bibfield{author}{\bibinfo{person}{Aurelien Decelle}, \bibinfo{person}{Florent
  Krzakala}, \bibinfo{person}{Cristopher Moore}, {and} \bibinfo{person}{Lenka
  Zdeborov{\'{a}}}.} \bibinfo{year}{2011}\natexlab{}.
\newblock \showarticletitle{Asymptotic analysis of the stochastic block model
  for modular networks and its algorithmic applications}.
\newblock \bibinfo{journal}{\emph{Physical Review E}} \bibinfo{volume}{84},
  \bibinfo{number}{6} (\bibinfo{year}{2011}).
\newblock
\urldef\tempurl%
\url{https://doi.org/10.1103/physreve.84.066106}
\showDOI{\tempurl}


\bibitem[\protect\citeauthoryear{Doreian}{Doreian}{1985}]%
        {Doreian-1985-structural}
\bibfield{author}{\bibinfo{person}{Patrick Doreian}.}
  \bibinfo{year}{1985}\natexlab{}.
\newblock \showarticletitle{Structural equivalence in a psychology journal
  network}.
\newblock \bibinfo{journal}{\emph{Journal of the American Society for
  Information Science}} \bibinfo{volume}{36}, \bibinfo{number}{6}
  (\bibinfo{year}{1985}), \bibinfo{pages}{411--417}.
\newblock
\urldef\tempurl%
\url{https://doi.org/10.1002/asi.4630360611}
\showDOI{\tempurl}


\bibitem[\protect\citeauthoryear{Eagle and Pentland}{Eagle and
  Pentland}{2005}]%
        {Eagle-2005-Reality}
\bibfield{author}{\bibinfo{person}{Nathan Eagle} {and}
  \bibinfo{person}{Alex~(Sandy) Pentland}.} \bibinfo{year}{2005}\natexlab{}.
\newblock \showarticletitle{Reality mining: sensing complex social systems}.
\newblock \bibinfo{journal}{\emph{Personal and Ubiquitous Computing}}
  \bibinfo{volume}{10}, \bibinfo{number}{4} (\bibinfo{year}{2005}),
  \bibinfo{pages}{255--268}.
\newblock
\urldef\tempurl%
\url{https://doi.org/10.1007/s00779-005-0046-3}
\showDOI{\tempurl}


\bibitem[\protect\citeauthoryear{Ghasemian, Hosseinmardi, and
  Clauset}{Ghasemian et~al\mbox{.}}{2018}]%
        {Ghasemian-2018-evaluating}
\bibfield{author}{\bibinfo{person}{Amir Ghasemian}, \bibinfo{person}{Homa
  Hosseinmardi}, {and} \bibinfo{person}{Aaron Clauset}.}
  \bibinfo{year}{2018}\natexlab{}.
\newblock \showarticletitle{Evaluating overfit and underfit in models of
  network community structure}.
\newblock \bibinfo{journal}{\emph{arXiv:1802.10582}} (\bibinfo{year}{2018}).
\newblock


\bibitem[\protect\citeauthoryear{Goel, Sharma, Wang, and Yin}{Goel
  et~al\mbox{.}}{2013}]%
        {Goel-2013-Discovering}
\bibfield{author}{\bibinfo{person}{Ashish Goel}, \bibinfo{person}{Aneesh
  Sharma}, \bibinfo{person}{Dong Wang}, {and} \bibinfo{person}{Zhijun Yin}.}
  \bibinfo{year}{2013}\natexlab{}.
\newblock \showarticletitle{Discovering similar users on {Twitter}}. In
  \bibinfo{booktitle}{\emph{11th Workshop on Mining and Learning with Graphs}}.
\newblock


\bibitem[\protect\citeauthoryear{Gupta, Goel, Lin, Sharma, Wang, and
  Zadeh}{Gupta et~al\mbox{.}}{2013}]%
        {Gupta-2013-WTF}
\bibfield{author}{\bibinfo{person}{Pankaj Gupta}, \bibinfo{person}{Ashish
  Goel}, \bibinfo{person}{Jimmy Lin}, \bibinfo{person}{Aneesh Sharma},
  \bibinfo{person}{Dong Wang}, {and} \bibinfo{person}{Reza Zadeh}.}
  \bibinfo{year}{2013}\natexlab{}.
\newblock \showarticletitle{{WTF}: the who to follow service at {Twitter}}. In
  \bibinfo{booktitle}{\emph{Proceedings of the 22nd international conference on
  World Wide Web}}. \bibinfo{publisher}{{ACM} Press}.
\newblock
\urldef\tempurl%
\url{https://doi.org/10.1145/2488388.2488433}
\showDOI{\tempurl}


\bibitem[\protect\citeauthoryear{Holme}{Holme}{2005}]%
        {Holme-2005-CP}
\bibfield{author}{\bibinfo{person}{Petter Holme}.}
  \bibinfo{year}{2005}\natexlab{}.
\newblock \showarticletitle{Core-periphery organization of complex networks}.
\newblock \bibinfo{journal}{\emph{Physical Review E}} \bibinfo{volume}{72},
  \bibinfo{number}{4} (\bibinfo{year}{2005}).
\newblock
\urldef\tempurl%
\url{https://doi.org/10.1103/physreve.72.046111}
\showDOI{\tempurl}


\bibitem[\protect\citeauthoryear{Kim and Leskovec}{Kim and Leskovec}{2011}]%
        {Kim-2011-completion}
\bibfield{author}{\bibinfo{person}{Myunghwan Kim} {and} \bibinfo{person}{Jure
  Leskovec}.} \bibinfo{year}{2011}\natexlab{}.
\newblock \showarticletitle{The Network Completion Problem: Inferring Missing
  Nodes and Edges in Networks}.
\newblock In \bibinfo{booktitle}{\emph{Proceedings of the SIAM Conference on
  Data Mining}}. \bibinfo{publisher}{Society for Industrial and Applied
  Mathematics}, \bibinfo{pages}{47--58}.
\newblock
\urldef\tempurl%
\url{https://doi.org/10.1137/1.9781611972818.5}
\showDOI{\tempurl}


\bibitem[\protect\citeauthoryear{Kleinberg}{Kleinberg}{2006}]%
        {Kleinberg-2006-ICM}
\bibfield{author}{\bibinfo{person}{Jon Kleinberg}.}
  \bibinfo{year}{2006}\natexlab{}.
\newblock \showarticletitle{Complex Networks and Decentralized Search
  Algorithms}. In \bibinfo{booktitle}{\emph{Proceedings of the International
  Congress of Mathematicians}}.
\newblock


\bibitem[\protect\citeauthoryear{Klimt and Yang}{Klimt and Yang}{2004}]%
        {Klimt-2004-Enron}
\bibfield{author}{\bibinfo{person}{Bryan Klimt} {and} \bibinfo{person}{Yiming
  Yang}.} \bibinfo{year}{2004}\natexlab{}.
\newblock \showarticletitle{{The Enron Corpus: A New Dataset for Email
  Classification Research}}.
\newblock In \bibinfo{booktitle}{\emph{Machine Learning: {ECML} 2004}}.
  \bibinfo{publisher}{Springer Berlin Heidelberg}, \bibinfo{pages}{217--226}.
\newblock
\urldef\tempurl%
\url{https://doi.org/10.1007/978-3-540-30115-8_22}
\showDOI{\tempurl}


\bibitem[\protect\citeauthoryear{Kossinets}{Kossinets}{2006}]%
        {Kossinets-2006-missing}
\bibfield{author}{\bibinfo{person}{Gueorgi Kossinets}.}
  \bibinfo{year}{2006}\natexlab{}.
\newblock \showarticletitle{Effects of missing data in social networks}.
\newblock \bibinfo{journal}{\emph{Social Networks}} \bibinfo{volume}{28},
  \bibinfo{number}{3} (\bibinfo{year}{2006}), \bibinfo{pages}{247--268}.
\newblock
\urldef\tempurl%
\url{https://doi.org/10.1016/j.socnet.2005.07.002}
\showDOI{\tempurl}


\bibitem[\protect\citeauthoryear{Laumann, Marsden, and Prensky}{Laumann
  et~al\mbox{.}}{1989}]%
        {Laumann-1989-boundary}
\bibfield{author}{\bibinfo{person}{Edward~O Laumann}, \bibinfo{person}{Peter~V
  Marsden}, {and} \bibinfo{person}{David Prensky}.}
  \bibinfo{year}{1989}\natexlab{}.
\newblock \showarticletitle{The boundary specification problem in network
  analysis}.
\newblock \bibinfo{journal}{\emph{Research methods in social network analysis}}
   \bibinfo{volume}{61} (\bibinfo{year}{1989}), \bibinfo{pages}{87}.
\newblock


\bibitem[\protect\citeauthoryear{Laumann and Pappi}{Laumann and Pappi}{1976}]%
        {Laumann-1976-collective}
\bibfield{author}{\bibinfo{person}{Edward~O. Laumann} {and}
  \bibinfo{person}{Franz~U. Pappi}.} \bibinfo{year}{1976}\natexlab{}.
\newblock \bibinfo{booktitle}{\emph{Networks of collective action: A
  perspective on community influence systems (Quantitative studies in social
  relations)}}.
\newblock \bibinfo{publisher}{Academic Press}.
\newblock


\bibitem[\protect\citeauthoryear{Leskovec, Kleinberg, and Faloutsos}{Leskovec
  et~al\mbox{.}}{2007}]%
        {Leskovec-2007-densification}
\bibfield{author}{\bibinfo{person}{Jure Leskovec}, \bibinfo{person}{Jon
  Kleinberg}, {and} \bibinfo{person}{Christos Faloutsos}.}
  \bibinfo{year}{2007}\natexlab{}.
\newblock \showarticletitle{Graph evolution: Densification and shrinking
  diameters}.
\newblock \bibinfo{journal}{\emph{{ACM} Transactions on Knowledge Discovery
  from Data}} \bibinfo{volume}{1}, \bibinfo{number}{1} (\bibinfo{year}{2007}),
  \bibinfo{pages}{2--es}.
\newblock
\urldef\tempurl%
\url{https://doi.org/10.1145/1217299.1217301}
\showDOI{\tempurl}


\bibitem[\protect\citeauthoryear{Liben-Nowell and Kleinberg}{Liben-Nowell and
  Kleinberg}{2007}]%
        {LibenNowell-2007-link-pred}
\bibfield{author}{\bibinfo{person}{David Liben-Nowell} {and}
  \bibinfo{person}{Jon Kleinberg}.} \bibinfo{year}{2007}\natexlab{}.
\newblock \showarticletitle{The link-prediction problem for social networks}.
\newblock \bibinfo{journal}{\emph{Journal of the American Society for
  Information Science and Technology}} \bibinfo{volume}{58},
  \bibinfo{number}{7} (\bibinfo{year}{2007}), \bibinfo{pages}{1019--1031}.
\newblock
\urldef\tempurl%
\url{https://doi.org/10.1002/asi.20591}
\showDOI{\tempurl}


\bibitem[\protect\citeauthoryear{Liben-Nowell, Novak, Kumar, Raghavan, and
  Tomkins}{Liben-Nowell et~al\mbox{.}}{2005}]%
        {LibenNowell-2005-geographic}
\bibfield{author}{\bibinfo{person}{D. Liben-Nowell}, \bibinfo{person}{J.
  Novak}, \bibinfo{person}{R. Kumar}, \bibinfo{person}{P. Raghavan}, {and}
  \bibinfo{person}{A. Tomkins}.} \bibinfo{year}{2005}\natexlab{}.
\newblock \showarticletitle{Geographic routing in social networks}.
\newblock \bibinfo{journal}{\emph{Proceedings of the National Academy of
  Sciences}} \bibinfo{volume}{102}, \bibinfo{number}{33} (\bibinfo{date}{aug}
  \bibinfo{year}{2005}), \bibinfo{pages}{11623--11628}.
\newblock
\urldef\tempurl%
\url{https://doi.org/10.1073/pnas.0503018102}
\showDOI{\tempurl}


\bibitem[\protect\citeauthoryear{L\"{u} and Zhou}{L\"{u} and Zhou}{2011}]%
        {Lu-2011-link-pred}
\bibfield{author}{\bibinfo{person}{Linyuan L\"{u}} {and} \bibinfo{person}{Tao
  Zhou}.} \bibinfo{year}{2011}\natexlab{}.
\newblock \showarticletitle{Link prediction in complex networks: A survey}.
\newblock \bibinfo{journal}{\emph{Physica A: Statistical Mechanics and its
  Applications}} \bibinfo{volume}{390}, \bibinfo{number}{6}
  (\bibinfo{year}{2011}), \bibinfo{pages}{1150--1170}.
\newblock
\urldef\tempurl%
\url{https://doi.org/10.1016/j.physa.2010.11.027}
\showDOI{\tempurl}


\bibitem[\protect\citeauthoryear{Mossel, Neeman, and Sly}{Mossel
  et~al\mbox{.}}{2014}]%
        {Mossel-2014-belief}
\bibfield{author}{\bibinfo{person}{Elchanan Mossel}, \bibinfo{person}{Joe
  Neeman}, {and} \bibinfo{person}{Allan Sly}.} \bibinfo{year}{2014}\natexlab{}.
\newblock \showarticletitle{Belief propagation, robust reconstruction and
  optimal recovery of block models}. In \bibinfo{booktitle}{\emph{Conference on
  Learning Theory}}. \bibinfo{pages}{356--370}.
\newblock


\bibitem[\protect\citeauthoryear{Rapoport}{Rapoport}{1953}]%
        {Rapoport-1953-triadic}
\bibfield{author}{\bibinfo{person}{Anatole Rapoport}.}
  \bibinfo{year}{1953}\natexlab{}.
\newblock \showarticletitle{Spread of information through a population with
  socio-structural bias {I}: {A}ssumption of transitivity}.
\newblock \bibinfo{journal}{\emph{Bulletin of Mathematical Biophysics}}
  \bibinfo{volume}{15}, \bibinfo{number}{4} (\bibinfo{date}{Dec.}
  \bibinfo{year}{1953}), \bibinfo{pages}{523--533}.
\newblock


\bibitem[\protect\citeauthoryear{Ratner, Bach, Ehrenberg, Fries, Wu, and
  R{\'{e}}}{Ratner et~al\mbox{.}}{2017}]%
        {Ratner-2017-Snorkel}
\bibfield{author}{\bibinfo{person}{Alexander Ratner},
  \bibinfo{person}{Stephen~H. Bach}, \bibinfo{person}{Henry Ehrenberg},
  \bibinfo{person}{Jason Fries}, \bibinfo{person}{Sen Wu}, {and}
  \bibinfo{person}{Christopher R{\'{e}}}.} \bibinfo{year}{2017}\natexlab{}.
\newblock \showarticletitle{Snorkel: rapid training data creation with weak
  supervision}.
\newblock \bibinfo{journal}{\emph{Proceedings of the {VLDB} Endowment}}
  \bibinfo{volume}{11}, \bibinfo{number}{3} (\bibinfo{year}{2017}),
  \bibinfo{pages}{269--282}.
\newblock
\urldef\tempurl%
\url{https://doi.org/10.14778/3157794.3157797}
\showDOI{\tempurl}


\bibitem[\protect\citeauthoryear{Ratner, De~Sa, Wu, Selsam, and R{\'e}}{Ratner
  et~al\mbox{.}}{2016}]%
        {Ratner-2016-data}
\bibfield{author}{\bibinfo{person}{Alexander~J Ratner},
  \bibinfo{person}{Christopher~M De~Sa}, \bibinfo{person}{Sen Wu},
  \bibinfo{person}{Daniel Selsam}, {and} \bibinfo{person}{Christopher R{\'e}}.}
  \bibinfo{year}{2016}\natexlab{}.
\newblock \showarticletitle{Data programming: Creating large training sets,
  quickly}. In \bibinfo{booktitle}{\emph{Advances in Neural Information
  Processing Systems}}. \bibinfo{pages}{3567--3575}.
\newblock


\bibitem[\protect\citeauthoryear{Rombach, Porter, Fowler, and Mucha}{Rombach
  et~al\mbox{.}}{2017}]%
        {Rombach-2017-CP}
\bibfield{author}{\bibinfo{person}{Puck Rombach}, \bibinfo{person}{Mason~A.
  Porter}, \bibinfo{person}{James~H. Fowler}, {and} \bibinfo{person}{Peter~J.
  Mucha}.} \bibinfo{year}{2017}\natexlab{}.
\newblock \showarticletitle{Core-Periphery Structure in Networks (Revisited)}.
\newblock \bibinfo{journal}{\emph{SIAM Rev.}} \bibinfo{volume}{59},
  \bibinfo{number}{3} (\bibinfo{year}{2017}), \bibinfo{pages}{619--646}.
\newblock
\urldef\tempurl%
\url{https://doi.org/10.1137/17m1130046}
\showDOI{\tempurl}


\bibitem[\protect\citeauthoryear{Romero, Uzzi, and Kleinberg}{Romero
  et~al\mbox{.}}{2016}]%
        {Romero-2016-stress}
\bibfield{author}{\bibinfo{person}{Daniel~M. Romero}, \bibinfo{person}{Brian
  Uzzi}, {and} \bibinfo{person}{Jon~M. Kleinberg}.}
  \bibinfo{year}{2016}\natexlab{}.
\newblock \showarticletitle{Social Networks Under Stress}. In
  \bibinfo{booktitle}{\emph{Proc. International World Wide Web Conference}}.
  \bibinfo{pages}{9--20}.
\newblock
\urldef\tempurl%
\url{https://doi.org/10.1145/2872427.2883063}
\showDOI{\tempurl}


\bibitem[\protect\citeauthoryear{Sarkar, Chakrabarti, et~al\mbox{.}}{Sarkar
  et~al\mbox{.}}{2015}]%
        {Sarkar-2015-consistency}
\bibfield{author}{\bibinfo{person}{Purnamrita Sarkar},
  \bibinfo{person}{Deepayan Chakrabarti}, {et~al\mbox{.}}}
  \bibinfo{year}{2015}\natexlab{}.
\newblock \showarticletitle{The consistency of common neighbors for link
  prediction in stochastic blockmodels}. In \bibinfo{booktitle}{\emph{Advances
  in Neural Information Processing Systems}}. \bibinfo{pages}{3016--3024}.
\newblock


\bibitem[\protect\citeauthoryear{Sarkar, Chakrabarti, and Moore}{Sarkar
  et~al\mbox{.}}{2011}]%
        {Sarkar-2011-justifications}
\bibfield{author}{\bibinfo{person}{Purnamrita Sarkar},
  \bibinfo{person}{Deepayan Chakrabarti}, {and} \bibinfo{person}{Andrew~W.
  Moore}.} \bibinfo{year}{2011}\natexlab{}.
\newblock \showarticletitle{Theoretical Justification of Popular Link
  Prediction Heuristics}. In \bibinfo{booktitle}{\emph{Proceedings of the
  Twenty-Second International Joint Conference on Artificial Intelligence}}.
  \bibinfo{publisher}{AAAI Press}, \bibinfo{pages}{2722--2727}.
\newblock
\urldef\tempurl%
\url{https://doi.org/10.5591/978-1-57735-516-8/IJCAI11-453}
\showDOI{\tempurl}


\bibitem[\protect\citeauthoryear{Watts and Strogatz}{Watts and
  Strogatz}{1998}]%
        {Watts-1998-small-world}
\bibfield{author}{\bibinfo{person}{Duncan~J. Watts} {and}
  \bibinfo{person}{Steven~H. Strogatz}.} \bibinfo{year}{1998}\natexlab{}.
\newblock \showarticletitle{Collective dynamics of `small-world' networks}.
\newblock \bibinfo{journal}{\emph{Nature}}  \bibinfo{volume}{393}
  (\bibinfo{year}{1998}), \bibinfo{pages}{440--442}.
\newblock


\bibitem[\protect\citeauthoryear{Yin, Benson, Leskovec, and Gleich}{Yin
  et~al\mbox{.}}{2017}]%
        {Yin-2017-local}
\bibfield{author}{\bibinfo{person}{Hao Yin}, \bibinfo{person}{Austin~R.
  Benson}, \bibinfo{person}{Jure Leskovec}, {and} \bibinfo{person}{David~F.
  Gleich}.} \bibinfo{year}{2017}\natexlab{}.
\newblock \showarticletitle{Local Higher-Order Graph Clustering}. In
  \bibinfo{booktitle}{\emph{Proceedings of the 23rd ACM SIGKDD International
  Conference on Knowledge Discovery and Data Mining}}.
  \bibinfo{publisher}{ACM}, \bibinfo{pages}{555--564}.
\newblock
\urldef\tempurl%
\url{https://doi.org/10.1145/3097983.3098069}
\showDOI{\tempurl}


\bibitem[\protect\citeauthoryear{Zhang, Martin, and Newman}{Zhang
  et~al\mbox{.}}{2015}]%
        {Zhang-2015-SBM-CP}
\bibfield{author}{\bibinfo{person}{Xiao Zhang}, \bibinfo{person}{Travis
  Martin}, {and} \bibinfo{person}{M.~E.~J. Newman}.}
  \bibinfo{year}{2015}\natexlab{}.
\newblock \showarticletitle{Identification of core-periphery structure in
  networks}.
\newblock \bibinfo{journal}{\emph{Physical Review E}} \bibinfo{volume}{91},
  \bibinfo{number}{3} (\bibinfo{year}{2015}).
\newblock
\urldef\tempurl%
\url{https://doi.org/10.1103/physreve.91.032803}
\showDOI{\tempurl}


\end{thebibliography}
\balance

\end{document}